\documentclass[11pt]{article}

\usepackage[letterpaper,margin=1in]{geometry}
\usepackage{amsmath}
\usepackage{amssymb}
\usepackage{amsthm}
\usepackage{fullpage}
\usepackage[unicode,colorlinks=true,citecolor=black,filecolor=black,linkcolor=black,urlcolor=black, pdfpagelabels, plainpages=false]{hyperref}
\usepackage{graphicx}
\usepackage{nicefrac}
\usepackage{xcolor}
\usepackage{wrapfig}
\usepackage{enumitem}
\usepackage{subcaption}
\newcommand{\set}[1]{\{#1\}}
\newcommand{\fnc}{f}
\newcommand{\ext}{f'}
\renewcommand{\Re}{\mathbb{R}}
\newcommand{\dist}{D}
\newcommand{\lip}{\mathrm{Lip}}
\newcommand{\norm}[1]{\left\|#1\right\|}%
\newcommand{\eps}{\varepsilon}

\newcommand{\Rbb}{\mathbb{R}}

\newtheorem{theorem}{Theorem}[section]

\newtheorem{claim}[theorem]{Claim}
\newtheorem{lemma}[theorem]{Lemma}
\newtheorem{corollary}[theorem]{Corollary}

\newtheorem{definition}[theorem]{Definition}

\newtheorem{openproblem}{Open Problem}

\iftrue
\makeatletter
\renewcommand{\paragraph}{%
  \@startsection{paragraph}{4}%
  {\z@}{2.ex \@plus 1ex \@minus .2ex}{-1em}%
  {\normalfont\normalsize\bfseries}%
}
\makeatother
\fi

\title{Nonlinear Dimension Reduction via Outer Bi-Lipschitz Extensions\footnote{An extended abstract appeared in the proceedings of STOC~2018.}}
\author{
Sepideh Mahabadi\\TTIC\\\small{mahabadi@ttic.edu}\and
Konstantin Makarychev \\ Northwestern University\\\small{konstantin@northwestern.edu}\and
Yury Makarychev\thanks{Supported by NSF awards CCF-1718820, CCF-1150062, and IIS-1302662.}\\TTIC\\\small{yury@ttic.edu} \and
Ilya Razenshteyn\\Microsoft Research\\\small{ilyaraz@microsoft.com}}
\date{}
\begin{document}
\maketitle
\begin{abstract}
We introduce and study the notion of an \emph{outer bi-Lipschitz extension} of a map between Euclidean spaces.
The notion is a natural analogue of the notion of a \emph{Lipschitz extension} of a Lipschitz map.
We show that for every map $f$ there exists an outer bi-Lipschitz extension $f'$ whose distortion is greater than that of
$f$ by at most a constant factor. This result can be seen as a counterpart of the classic Kirszbraun theorem for outer
bi-Lipschitz extensions. We also study outer bi-Lipschitz extensions of near-isometric maps and show
upper and lower bounds for them.
Then, we present applications of our results to prioritized and terminal dimension reduction problems.
\begin{itemize}
\item We prove a \emph{prioritized} variant of the Johnson--Lindenstrauss lemma: given a set of points $X\subset \Re^d$ of size $N$
and a permutation (``priority ranking'') of $X$, there exists an embedding $f$ of $X$ into $\Rbb^{O(\log N)}$ with distortion $O(\log \log N)$
such that the point of rank $j$ has only $O(\log^{3 + \eps} j)$ non-zero coordinates -- more specifically, all but the first $O(\log^{3+\eps} j)$ coordinates are equal to~$0$; the distortion of $f$ restricted to the first $j$ points (according to the ranking)
is at most $O(\log\log j)$. The result makes a progress towards answering an open question by Elkin, Filtser, and Neiman about prioritized dimension reductions.
\item We prove that given a set $X$ of $N$ points in $\Re^d$, there exists a \textit{terminal} dimension reduction embedding
of $\Rbb^d$ into $\Re^{d'}$, where $d' = O(\frac{\log N}{\eps^4})$, which preserves distances $\|x-y\|$ between points $x\in X$ and $y \in \Re^{d}$, up to a multiplicative factor of $1 \pm \eps$.
This improves a recent result by Elkin, Filtser, and Neiman.
\end{itemize}
The dimension reductions that we obtain are nonlinear, and this nonlinearity is necessary.
\end{abstract}

\section{Introduction}
In this paper, we introduce and study the notion of \textit{an outer bi-Lipschitz extension}.
The notion is a natural analogue of the notion of \text{a Lipschitz extension}, which is widely used in mathematics and theoretical
computer science. Recall that a map $f: X \to Y$ is $C$-Lipschitz if for any two points $x, y\in X$ we have $d_Y(f(x),f(y))\leq C \cdot d_X(x,y)$;
the Lipschitz constant of $f$ is the minimum $C$ such that $f$ is $C$-Lipschitz.
In the Lipschitz extension problem, given a Lipschitz map $f$ from a subset $A$ of $X$ to $Y$ and a superset $A'\supset A$, the goal is to find an extension map $f'$ from $A'$ to $Y$ such that the Lipschitz constant of $f'$ is equal to or not
significantly larger than the Lipschitz constant of $f$.
This problem has found numerous applications in mathematics and theoretical computer science
(see e.g., \cite{K34, M34, MP84, JL84, LN05, MN06, NPSS06, AKR15, MM10, MM16}).
One of the most important results in the field is the Kirszbraun theorem, which states that any map $f:A\to \Rbb^m$
from a subset $A$ of Euclidean space $\Rbb^n$ to Euclidean space $\Rbb^m$ can be extended to a map $f':\Rbb^n\to \Rbb^m$ so
that the Lipschitz constant of $f'$ equals that of $f$~\cite{K34} (see Theorem~\ref{thm:Kirszbraun} in Section~\ref{sec:prelim}; see also~\cite{AT08}).

\paragraph{Outer bi-Lipschitz extension.} In this paper, we prove several analogues of the Kirszbraun theorem for \emph{bi-Lipschitz} maps. The bi-Lipschitz constant of a map $f\colon X \to Y$ is the minimum $D$ such that for some $\lambda > 0$ and every $x, y \in X$,
$
\lambda \cdot d_X(x,y) \leq d_Y(f(x),f(y)) \leq \lambda \cdot D \cdot d_X(x,y)
$.
If there is no such number $D$, we say that the map is not bi-Lipschitz. Bi-Lipschitz maps are also known as embeddings with distortion $D$. Low distortion metric embedding have numerous applications in approximation and online algorithms (see e.g. \cite{LLR95,AR98,Bartal98,Feige98,M02,ABN06,ANCDGKNS05,Bartal06,CMM06,FRT08,ALN08,KMM11, MMV12, MMV14, MMSJ16,EFN17}); hardness of approximation (see e.g.~\cite{KV15}); computational geometry (see e.g. \cite{M02, IM04} and references therein); and sketching, streaming, and similarity search algorithms (see e.g.~\cite{IM98, CMS01,Achlioptas03,BES06,CK06,NS07,OR07,AIK08,AIK09,N14,ANNRW17,andoni2018holder}).

Since bi-Lipschitz  maps are widely used in mathematics and theoretical computer science, it is natural to ask whether there is a counterpart of the Kirszbraun theorem for bi-Lipschitz maps.
\begin{quote}
Given a bi-Lipschitz map $f$ from a subset of $\Rbb^n$ to $\Rbb^m$, can we extend it to a bi-Lipschitz map
from the whole space $\Rbb^n$ to $\Rbb^m$?
\end{quote}

This question has been extensively studied in the literature (see e.g.~\cite{Ghamsari,PV93,Vaisala94,ATV03,AT09,Kovalev17}). It turns out that the answer to this question depends on the geometry of the set $A$. In general, the answer is ``no''.
For instance, consider a map that maps points $0$, $1$, $2$ to $0$, $-1$, $2$, respectively.
There is no continuous one-to-one extension of this map to $\Rbb$, let alone a bi-Lipschitz extension.
The reason is that in one dimension we cannot connect points $0$ and $-1$ and points $-1$ and $2$ with non-intersecting paths. However, we can easily do this in $\Rbb^2$.
This observation suggests the following idea. Let $A\subset \Rbb^n$ and $f:A \to \Rbb^m$ be a bi-Lipschitz map. Let us allow extension $f'$ of $f$ to use additional dimensions or, in other words, allow $f'$ to map points $x \in \Rbb^n\setminus A$ to points in some higher-dimensional (ambient) space $\Rbb^{m'}$ that contains $\Rbb^m$. We get the following definition.
\begin{definition}[Outer extension]
A map $f': A'\to \Rbb^{m'}$ (where $m' \geq m$) is an outer extension of $f$
if $f(a) = f'(a)$ for all $a\in A$; we assume that $\Rbb^m$ is the subspace of $\Rbb^{m'}$ spanned by the first
$m$ standard basis vectors; that is,
we identify points $(x_1,\dots,x_m) \in \Rbb^m$ and $(x_1,\dots, x_m, 0, \dots, 0)\in\Rbb^{m'}$.
We say that the extension is proper if $m = m'$.
\end{definition}
%That is, the outer extension is allowed to use extra coordiantes/dimensions.
Note that the exact dimension of the image is not very important in many applications in computer science,
as long as the dimension is comparable to $m$ and $n$. Therefore,  outer extensions seem to be as useful as proper (standard) extensions.
However, in stark contrast with proper bi-Lipschitz extensions, outer bi-Lipschitz extensions \textit{always} exist -- for
every bi-Lipschitz map $f: A \to \Rbb^m$ there exists an outer bi-Lipschitz extension $f':\Rbb^n \to \Rbb^{m'}$, as we prove in this paper.

%%%%%%%%%%%%%%%%%%%%%%%
\subsection{Results}
\paragraph{Outer bi-Lipschitz Extensions.}
One of the main results of this paper is an analogue of the Kirszbraun theorem for bi-Lipschitz maps.
\begin{theorem}\label{lem:bi-lip-ext}
Let $X \subset \Rbb^n$ and $\fnc : X\rightarrow \Re^m$ be a bi-Lipschitz map with distortion at most $\dist$.
There exists an outer extension $\ext: \Re^n \rightarrow \Re^{m'}$ of $f$ with
the distortion at most $3\dist$ and $m'=n+m$.
\end{theorem}

The  main difference between the outer bi-Lipschitz extension from Theorem~\ref{lem:bi-lip-ext} and the Lipschitz extension from the Kirszbraun theorem --
aside from the difference we discussed above (that Theorem~\ref{lem:bi-lip-ext} gives an \textit{outer} extension and not a proper extension) --
is that while the Lipschitz extension preserves the Lipschitz constant of the map exactly, the bi-Lipschitz extension preserves the distortion only up to a constant factor. This limitation is unavoidable; it is easy to see that even in the example we considered --
extending the map $f$ that sends $0$, $1$, $2$ to $0$, $-1$, $2$, respectively -- the distortion of any outer extension of $f$ is greater than the distortion of $f$.
Thus, for arbitrary bi-Lipschitz maps we cannot get a result stronger than Theorem~\ref{lem:bi-lip-ext} (except that factor $3$ in the statement of the theorem can be potentially replaced with a smaller factor $c > 1$).

We then focus on an important class of near-isometric maps, maps with distortion $D = 1 + \varepsilon$.
Observe that if the distortion of $f$ is exactly $1$ (i.e., $f$ is an isometric embedding), it can be extended to an isometric embedding of the whole space $\Rbb^n$ into $\Rbb^{m'}$.
In this case, we can extend $f$ without increasing its distortion.
What happens if the distortion of $f$ is close to $1$ but not $1$?
Let $\varphi(\eps)$ be the smallest $\varepsilon'$ such that the following holds:
for every map $f:A \to \Rbb^m$ with distortion at most $D = 1 + \varepsilon$, there exists
an outer extension $f':\Rbb^n\to \Rbb^{m'}$ with distortion at most $D' = 1 + \varepsilon'$.
Note that $\varphi(0) = 0$, as discussed above.
\begin{openproblem}
Find the asymptotic behavior of $\varphi(\eps)$ as $\eps \to 0$. Does $\varphi(\eps) \to 0$ as $\eps \to 0$?
\end{openproblem}
\noindent We study this problem and get partial results for it. First, we show that $\varphi(\eps) \geq \Omega(1/\log^2 (1/\eps))$.
\begin{theorem}\label{thm:lower-bound-on-phi}
There exists a map $f: X \to \Re$, where $X\subset \Re$, with the distortion $1+\eps$, such that every outer extension $f': \Re \to \Re^m$ of $f$ has distortion at least $1+\Omega(\frac{1}{\log^2 (1/\eps)})$.
\end{theorem}
\noindent Note that $1/\log^2 (1/\eps) \to 0$ as $\eps \to 0$, but the dependence of $1/\log^2 (1/\eps)$ on $\eps$ is not polynomial and, in our opinion,
highly unusual. This result rules out the possibility that $\varphi(\eps) = O(\eps^{1/k})$ for any $k$.
Further, we provide some evidence that $\varphi(\eps)$ might, in fact, be equal to $1 + \Theta(\frac{1}{\log^2 (1/\eps)})$.
Namely, we prove the following result for 1-dimensional case: for every map from $X\subset \Rbb$ to $\Rbb$, there is an outer extension with $D' = 1 + O(\frac{1}{\log^2 (1/\eps)})$. By Theorem~\ref{thm:lower-bound-on-phi}, this bound is asymptotically optimal.
\begin{theorem}\label{thm:line-embedding}
Let $X\subset \Re$ and $f: X \to \Re$ be a map with the distortion at most $1+\eps$. There exists an outer extension $f'\colon \Re \to \Re^2$ of $f$ with the distortion at most $1+O(\frac{1}{\log^2 (1/\eps)})$.
\end{theorem}

We also consider a simpler problem of extending a near-isometric map by one point. We prove the following result.
\begin{theorem} \label{thm:one-point-ext}
Let $f$ be a $(1+\varepsilon)$-bi-Lipschitz map from a subset $X$ of $\Rbb^n$ to $\Rbb^m$
and $u \in \Rbb^n$. There exists an outer extension $f':X\cup\set{u} \to \Rbb^{m+1}$ of $f$ with the distortion at most
$1 + O(\sqrt{\varepsilon})$.
\end{theorem}
The bound in this theorem is asymptotically tight -- there exist a map $f$ from a subset of $\Rbb$ to $\Rbb$ and a point $u\in \Rbb$ such that every outer extension of $f$ to  $u$  has distortion $1+\Omega(\sqrt \eps)$.

\paragraph{Computability.} Given sets $A\subset A' \subset \Rbb^n$ and a map $f:A\to \Rbb^m$,
we can compute an outer extension $f': A' \to \Rbb^n$ with the least possible distortion using
semidefinite programming (SDP). The running time is polynomial in $|A'|$ and $\log 1/\delta$, where $\delta$ is the desired precision.
In particular, we can find outer extensions $f'$, whose existence is guaranteed by Theorems~\ref{lem:bi-lip-ext} and~\ref{thm:one-point-ext}.

\paragraph{Applications.} Using our extension results, we obtain \textit{prioritized} and \textit{terminal} dimension reductions~\cite{EFN15,EFN17}.
Recall the statement of the Johnson--Lindenstrauss lemma~\cite{JL84}.
\begin{theorem}[The Johnson--Lindenstrauss Lemma~\cite{JL84}]
For every $0 < \eps < 1/2$ and every set $X \subset \Rbb^d$ of size $N$,
there exists an embedding $f: X \to \Rbb^{d'}$, where $d' = O\left(\frac{\log N}{\eps^2}\right)$,
such that for every $p, q \in X$: $
\|p - q\|_2 \leq \|f(p) - f(q)\|_2 \leq (1 + \eps)  \|p - q\|_2
$.
\end{theorem}
Prioritized metric structures and embeddings were introduced and studied by Elkin, Filtser, and Neiman~\cite{EFN15}. Among several very interesting results obtained in~\cite{EFN15},
one is a construction of prioritized embeddings.  We give a definition of a prioritized dimension reduction in the spirit of~\cite{EFN15}.
\begin{definition}[Prioritized dimension reduction]\label{def:prioritizeddimered}
Consider a set of points $X \subset \Re^d$ of size $N$. Let $\pi$ be a bijection from $[N] = \set{1,\dots,N}$
to $X$, which defines a \emph{priority ranking} of $X$: $\pi(1), \dots, \pi(N)$.
An embedding $f:X \to \Rbb^{d'}$ is an $(\alpha,\beta)$-prioritized dimension reduction, where $\alpha: [N] \to \Rbb$ and
$\beta: [N] \to {\mathbb N}$, if
\begin{itemize}
\item for every $j\in [N]$, the distortion of $f$ restricted to points $\pi(1), \dots, \pi(j)$ is at most $\alpha(j)$.
\item for every $j\in [N]$, $\pi(j)$ is mapped to a point $f(\pi(j))$ in $R^{\beta(j)}$; that is,
all but the first $\beta(j)$ coordinates of $f(\pi(j))$ are equal to $0$.
\end{itemize}
\end{definition}
Note that points $f(\pi(1)), \dots, f(\pi(j))$ lie in Euclidean space of dimension
$\beta(j)$ and $\beta(j)$ may potentially be much smaller than $\log N$ (when $j \ll N$).
The definition requires that the distortion of the distance between points $\pi(i)$ and $\pi(j)$ be at most
$\alpha(\max(i, j))$ (note that this condition is weaker than a similar condition in the definition
of \textit{a prioritized embedding} in \cite{EFN15}, which requires that the distortion be at most $\alpha(\min(i, j))$).

Ideally, we want to have a dimension reduction with parameters $(1+\eps, \mathrm{polylog}\ j)$.
\begin{openproblem}[{\cite[talk and pers.~comm.]{EFN15}}]
\label{open_problem_jl}
Is there a prioritized dimension reduction with parameters $(1+\eps, \mathrm{polylog}\ j)$?
\end{openproblem}
Very little is known about prioritized dimension reductions. The only known result follows from
Theorem 15 in \cite{EFN15}. (The theorem is a prioritized variant of Bourgain's theorem~\cite{B85} and is more general
than its corollary stated below.)
\begin{theorem}[\cite{EFN15}]
For every set $X\subset \Re^d$ and $\varepsilon > 0$, there is a $(c_1\log^{4 + \varepsilon} j, c_2\log^4 j)$-prioritized dimension reduction $f:X \to \Rbb^{O(\log^2 |X|)}$ (where $c_1, c_2$ depend only on $\varepsilon$).
\end{theorem}
We make further progress towards solving Open Problem~\ref{open_problem_jl}.
\begin{theorem}\label{thm:prioritized-embedding}
For every set $X\subset \Re^d$, $\varepsilon > 0$, and $N = |X|$, there exist
\begin{itemize}
\item a $(c_1\log_2\log_2 j, c_2\log_2^{3 + \varepsilon} j)$-prioritized dimension reduction $f:X \to \Rbb^{O\left(\frac{\log N}{\varepsilon^2}\right)}$, where $c_1 = 3 + \varepsilon$ and  $c_2 = O(1/\varepsilon^2)$,
\item a $((3+\eps)^k, c_1 \log_2 j \log^{1/k} N)$-prioritized dimension reduction $f:X \to \Rbb^{O\left(\frac{\log N}{\varepsilon^2}\right)}$ for every integer parameter $k > 1$, where $c_1 = O(1/\varepsilon^2)$.
\end{itemize}
The dimension reductions can be computed in polynomial time.
\end{theorem}
The first result gives a prioritized dimension reduction with a reasonably small distortion $O(\log\log j)$ and desired
polylogarithmic dimension. The second result gives a constant distortion and maps the first $j$ points to a subspace of dimension
$O(\log_2 j \log^{1/k} N)$.

Now we switch to another problem introduced by Elkin, Filtser, and Neiman~\cite{EFN17}.
\begin{definition}[Terminal dimension reduction] Suppose that we are given a set of points (which we call terminals)
$X\subset \Rbb^d$. We say that a map $f: \Rbb^d \to \Rbb^{d'}$ is a terminal dimension reduction with distortion $D$ if for every terminal $x \in X$ and point $p \in \Rbb^d$ ($p$ may be a terminal), we have
$$
\|p - x\| \leq \|f(p) - f(x)\| \leq D\,  \|p - x\|.
$$
\end{definition}
Elkin, Filtser, and Neiman~\cite{EFN17} proved that there exists a terminal dimension reduction with distortion $O(1)$ and
dimension $d' = O(\log |X|)$. We show how to obtain the distortion of $1 + \varepsilon$.
\begin{theorem}\label{thm:embedding-with-terminals}
For every set $X \subset \Rbb^d$ of size $N$ and parameter $0 < \eps < 1/2$,
there exists a terminal dimension reduction $f: X \to \Rbb^{d'}$ with distortion $1+\eps$, where $d' = O\left(\frac{\log N}{\eps^4}\right)$.
The dimension reduction can be computed in polynomial time.
\end{theorem}
It is an interesting question if the dimension $O\left(\frac{\log N}{\eps^4}\right)$ can be lowered. Since $f$ is also a (standard) dimension reduction for $X$, $d'$ must be at least $\Omega\left(\frac{\log N}{\eps^2}\right)$ as was shown by Larsen and Nelson~\cite{LN17} (see also ~\cite{AK17,LN16,Alon09}).
\begin{openproblem}
\label{open_terminal}
Is it possible to decrease the dimension to $O\left(\frac{\log N}{\eps^2}\right)$ in Theorem~\ref{thm:embedding-with-terminals}?
\end{openproblem}
After the conference version of this paper appeared, Open Problem~\ref{open_terminal} was resolved in the positive by Narayanan and Nelson~\cite{narayanan2018optimal}.

It is interesting that while most dimension reduction constructions described in the literature
are given by linear transformations, prioritized and terminal dimension reductions must be non-linear. In particular, all dimension reductions presented in this paper are non-linear.

%new
\paragraph{Prior Work on Outer bi-Lipschitz Extensions.} After the conference version of this paper was published,
Kovalev informed us about a relevant result by Alestalo and V\"ais\"al\"a~\cite[Theorem 5.5]{alestalo1997uniform}.
Proved in a different context, it states that every map $f$ with distortion $D$ has an outer bi-Lipschitz extension $f'$
with distortion at most $D' = \sqrt{7} D^2$. The result and its proof are similar to the statement and proof
of Theorem~\ref{lem:bi-lip-ext}. However, in Theorem~\ref{lem:bi-lip-ext}, the dependence of $D'$ on $D$ is \emph{linear.}

%merged from the conference version
We note that Makarychev and Makarychev~\cite{MM16} introduced a related notion of an external bi-Lipschitz extension, but that notion is significantly different from and less natural than the notion of the outer bi-Lipschitz extension studied in this paper.

\paragraph{Roadmap.}
In Section~\ref{sec:bi-lip-ext}, we prove Theorem~\ref{lem:bi-lip-ext}.  In Section~\ref{sec:one-point}, we obtain an optimal bound on one-point outer bi-Lipschitz extensions (prove Theorem~\ref{thm:one-point-ext} and show its optimality). Then, in Section~\ref{sec:applications}, we present applications of our results. Finally, in Section~\ref{sec:line-overview}, we give an overview of the proof of Theorem~\ref{thm:line-embedding}; we present the entire proof, as well as  a matching lower bound, in Section~\ref{sec:line-case}.

\subsection{Preliminaries}\label{sec:prelim}
In this paper, $\Rbb^n$ denotes $n$-dimensional Euclidean space, equipped with the standard Euclidean norm $\|\cdot\|$.
For $m < m'$, we identify $\Rbb^m$ with the $m$-dimensional subspace of $\Rbb^{m'}$ spanned by the first $m$ standard basis vectors
(in other words, we identify vectors $(x_1,\dots, x_m)\in\Rbb^m$ and $(x_1,\dots, x_m, 0, \dots, 0)\in \Rbb^{m'}$).

\begin{definition}[Lipschitz constant and distortion] Let $(X, d_X)$ and $(Y, d_Y)$ be metric spaces, and let $\fnc: X\rightarrow Y$ be a map. Define the Lipschitz constant of $f$ as
$
\norm{\fnc}_\lip = \mathrm{sup}_{x,y\in X} \frac{d_Y(\fnc(x),\fnc(y))}{d_X(x,y)}
$.
We say that the map $\fnc$ is Lipschitz if $\norm{\fnc}_\lip < \infty$. A map $\fnc$ is {non-expanding} if  $\norm{\fnc}_\lip\leq 1$.
The distortion or  bi-Lipschitz constant of an injective map $\fnc$ is $\dist = \dist(\fnc) = \norm{\fnc}_\lip\cdot\|\fnc^{-1}\|_\lip$.
If a map is not injective, its distortion is infinite. A map $\fnc$ is bi-Lipschitz if $\dist(\fnc) <\infty$.
\end{definition}

\begin{theorem}[Kirszbraun Extension Theorem]\label{thm:Kirszbraun}
Consider Euclidean spaces $\Rbb^n$ and $\Rbb^m$, and an arbitrary non-empty subset $X$ of $\Rbb^n$
% their arbitrary non-empty subsets $A \subset \Rbb^n$ and $B \subset \Rbb^m$.
Let $\fnc : X\rightarrow \Re^m$ be a Lipschitz map. There exists a proper extension $\fnc':\Re^n\rightarrow \Re^m$ of $f$ with the same Lipschitz constant as $f$: $\norm{\fnc'}_\lip = \norm{\fnc}_\lip$.
\end{theorem}

%............................
\section{Outer bi-Lipschitz extension}\label{sec:bi-lip-ext}
In this section, we prove Theorem~\ref{lem:bi-lip-ext} that states that any bi-Lipschitz map $f$ from a subset $X$ of $\Rbb^n$ to $\Rbb^m$ can be extended
to a bi-Lipschitz map $f':\Rbb^n \to \Rbb^{m'}$ for some $m' > m$. 
The result can be seen as a counterpart of the Kirszbraun theorem.

\iffalse
\begin{theorem}[Theorem~\ref{lem:bi-lip-ext}]
Let $X \subset \Rbb^n$ and $\fnc : X\rightarrow \Re^m$ be a bi-Lipschitz map with distortion at most $\dist$.
There exists an outer extension $\ext: \Re^n \rightarrow \Re^{m'}$ of $f$ with
the distortion at most $3\dist$ and $m'=n+m$.
\end{theorem}
\fi
\medskip
\noindent\textit{Informal overview of the proof idea.} For simplicity, let us assume for now that $f$ is near-isometric (it approximately preserves distances).
We want to construct a map $f':\Re^n \to \Re^{m'}$ that satisfies the following conditions:
\begin{enumerate}[label=(\arabic*)]
\item $f'$ is an outer extension of $f$; that is, $f'(x) = f(x)$ for every $x\in X$;
\item $\|f'(x) - f'(y)\| \leq O(\|x - y\|)$ for all $x,y\in \Rbb^n$ ;
\item $\|f'(x) - f'(y)\| \geq \Omega(\|x-y\|)$ for all $x,y\in \Rbb^n$ .
\end{enumerate}
First, using the Kirszbraun theorem, we find a Lipschitz extension $\tilde f: \Rbb^n \to \Rbb^m$.
If we were to let $f' =\tilde f$, then $f'$ would satisfy conditions (1) and (2) but not necessarily (3);
namely, for some points $x,y\in \Rbb^n$, the distance between $f'(x)$ and $f'(y)$ would potentially be considerably smaller than that between $x$ and $y$;
in fact, it could happen that $\tilde f(x) = \tilde f(y)$ for some $x\neq y$. Instead, we are going to let
$f'(x) = \tilde f(x) \oplus h(x)\in \Rbb^{n+m}$ for some map $h$ from $\Re^n$ to $\Re^{n}$. We will choose $h$ which satisfies the following conditions:
\begin{enumerate}[label=($\arabic*'$)]
\item For $x\in X$, $h(x) = 0$. This condition is necessary to ensure that $f'$ is an outer extension of $f$.
\item For all $x,y\in \Rbb^n$, $\|h(x) - h(y)\| \leq O(\|x-y\|)$ and thus $\|f'(x) - f'(y)\| \leq \|\tilde f(x) - \tilde f(y)\| + \|h(x) - h(y)\|
\leq O(\|x-y\|)$.
\item If $\|\tilde f(x) - \tilde f(y)\| \ll \|x-y\|$ for some $x,y\in \Rbb^n$, then $\|h(x) - h(y)\| = \Omega(\|x-y\|)$ and thus
$\|f'(x) - f'(y)\| \geq \|h(x) - h(y)\| \geq \Omega(\|x-y\|)$.
\end{enumerate}
As we see, if $h$ satisfies conditions ($1'$), ($2'$), and ($3'$), then $f' = \tilde f \oplus h$ satisfies conditions (1), (2), and (3).
Now we proceed with a formal proof. Our main task will be to define $h$ appropriately.
\begin{proof}
As above, let $\tilde f:\Re^n \to \Re^{m}$ be a Lipschitz extension of $f$ with $\|\tilde f\|_{\lip} = \|f\|_{\lip}$.
Further, let $g = f^{-1}:f(X)\to X$ be the inverse map of $f$ and $\tilde g: \Re^m\to \Re^n$ be its Lipschitz extension given by the Kirszbraun theorem.
Denote $\alpha = \norm{g}_{\lip}$. Since the distortion of $f$ is at most $D$,
$$\|f\|_{\lip} \leq D/\alpha, \qquad \|\tilde f\|_{\lip} \leq D/\alpha, \qquad \norm{g}_{\lip} \leq \alpha, \qquad \norm{\tilde g}_{\lip} \leq \alpha, \qquad \|\tilde g \circ \tilde f\|_{\lip}\leq D.$$
Let $ h(x)= \frac{\tilde g(\tilde f(x)) - x}{\sqrt2 \alpha}$ and
$
\fnc'(x) = \tilde f(x) \oplus h(x) = \tilde f(x) \oplus \frac{\tilde g(\tilde f(x)) - x}{\sqrt2 \alpha} \in \Rbb^{n+m}
$.
We verify that $f'$ satisfies conditions (1), (2), and (3) described in the proof overview above.

\medskip
\noindent\textbf{Condition (1).} We prove that $\fnc'$ is an outer extension of $\fnc$; i.e., for every $x\in X$, we have
$$
\fnc'(x) =\fnc(x) \oplus \frac{\tilde g(\tilde f(x)) - x}{\sqrt2 \alpha} =
\fnc(x) \oplus \frac{g(f(x)) - x}{\sqrt2 \alpha} = f(x) \oplus 0 = \fnc(x).$$

\medskip
\noindent\textbf{Condition (2).}  For every $x,y\in \Rbb^n$, we have
$$\sqrt 2 \alpha \cdot \|h(x) - h(y)\| = \|(x - \tilde g\circ \tilde f(x)) - (y - \tilde g\circ \tilde f(y))\|)\| \leq  \|x - y\| +  \|\tilde g\circ\tilde f(x) - \tilde g\circ\tilde f(y)\| \leq (1 + D) \|x - y\|.$$
Thus,
$$\|f'(x)-f'(y)\|^2 \leq \|\tilde f(x) - \tilde f(y)\|^2 + \|h(x) - h(y)\|^2 \leq \left(\Bigl(\frac{D}{\alpha}\Bigr)^2 + \frac{(1+D)^2}{2\alpha^2}\right) \|x-y\|^2 \leq \frac{3D^2}{\alpha^2}.$$
\iffalse
Denote by $id:\Rbb^n\to\Rbb^n$ the identity map.
\[
\|\fnc'\|_{\lip}^2 \leq  \norm{\fnc}_{\lip}^2 + \frac{1}{2\alpha^2} \norm{2\alpha^2 h}_{\lip}^2 \leq
\norm{\fnc}_{\lip}^2 + \frac{1}{2\alpha^2} \bigl(\norm{id}_{\lip} + \|\tilde f \circ \tilde g\|_{\lip}\bigr)^2
\leq \left(\frac{D}{\alpha}\right)^2 + \frac{\left(1 + D\right)^2}{2\alpha^2} \leq \frac{3D^2}{\alpha^2}
.
\]
\fi
Therefore, $\|f'\|_{\lip} \leq \sqrt{3} D/\alpha$.

\medskip
\noindent\textbf{Condition (3).} Finally, we prove that the Lipschitz constant of the inverse map $f'^{-1}$ is at most
$\sqrt{3}\alpha$.
Consider two distinct points $x, y\in \Rbb$. Let $\rho = \frac{\alpha\|\tilde f(x) - \tilde f(y)\|}{\|x-y\|}$.
If $\rho \geq 1$, then $\|f'(x) - f'(y)\| \geq \|\tilde f(x) - \tilde f(y)\| \geq \|x-y\| / \alpha$.
Otherwise, $\|\tilde g(\tilde f(x)) - \tilde g(\tilde f(x))\| \leq \rho \|x-y\| < \|x - y\|$, and
\begin{align*}
\|f'(x) - f'(y)\|^2 & = \|\tilde f(x) - \tilde f(y)\|^2 + \frac{1}{2\alpha^2} \|(x - y) - (\tilde g(\tilde f(x)) - \tilde g(\tilde f(y))) \|^2 \\ &\geq \frac{\rho^2}{\alpha^2} \|x - y\|^2 + \frac{(1 - \rho)^2 \|x-y\|^2}{2\alpha^2} =
\frac{(1 - 2\rho + 3\rho^2) \|x-y\|^2}{2\alpha^2} \geq \frac{\|x-y\|^2}{3\alpha^2}.
\end{align*}
Here we used that the minimum of the quadratic polynomial $1 - 2\rho + 3\rho^2$ equals $2/3$.
In both cases, we have $\|f'(x) - f'(y)\| \geq \nicefrac{\|x-y\|}{\sqrt3\alpha}$. Therefore,
$\|f'^{-1}\|_{\lip} \leq \sqrt{3} \alpha$. We conclude that the distortion of $\fnc'$ is at most $3D$.
\end{proof}

\section{One-point extension of near-isometric maps}\label{sec:one-point}
\subsection{Upper bound}
\newcommand{\Cconst}{3}

In this section, we prove Theorem~\ref{thm:one-point-ext}. The theorem states that every near-isometric map can be extended
to an extra point so that the extended map is also near isometric.
\iffalse
\begin{theorem}[Theorem~\ref{thm:one-point-ext}]
Let $f$ be a $(1+\varepsilon)$-bi-Lipschitz map from a subset $X$ of $\Rbb^n$ to $\Rbb^m$
and $u \in \Rbb^n$. There exists an outer extension $f':X\cup\set{u} \to \Rbb^{m+1}$ of $f$ with distortion at most
$1 + O(\sqrt{\varepsilon})$.
\end{theorem}
\fi
\begin{proof}[Proof of Theorem~\ref{thm:one-point-ext}]
Without loss of generality, we can make several simplifying assumptions.
First, it is sufficient to prove the theorem only for finite subsets $X$ of $\Rbb^n$; the statement for infinite subsets follows
from a simple compactness argument. We will assume that $\varepsilon \in (0,1)$, if $\varepsilon > 1$, the theorem follows
from Theorem~\ref{lem:bi-lip-ext}. Further, by rescaling $f$, if necessary, we may assume that $\|v - w\| \leq \|f(v) - f(w)\| \leq (1+ \varepsilon) \|v - w\|$ for every $v,w \in X$. In particular,
\begin{equation}\label{eq:bi-Lip}
\|v - w\|^2 \leq \|f(v) - f(w)\|^2 \leq (1+ 3\varepsilon) \|v - w\|^2.
\end{equation}
If $u \in X$ then there is nothing to prove, so we assume that $u\notin X$.
Let $v_0$ be the point closest to $u$ in $X$ (or one of the closest points to $u$ if there is more than one such point).
To simplify notation, we assume that $v_0 = 0$, $f(v_0) = 0$, and $\|v_0 - u\| = 1$.
Then $\|u\| = 1$ and $\|u - v\| \geq 1$ for every $v\in X$. The theorem will follow from the following lemma.
\begin{lemma}\label{lem:one-point-ext-main}
There exists a vector $u'\in \Rbb^m$ such that
\begin{enumerate}
\item $\|u'\| \leq 1$,
\item $|\langle u', f(v)\rangle - \langle u, v\rangle| \leq \Cconst\sqrt{\varepsilon}\,(\|v\|^2 + 1)$ for every $v\in X$.
\end{enumerate}
\end{lemma}
\begin{proof}
Let $\Lambda = \set{\lambda \in \Rbb^X: \|\lambda\|_1 \leq 1}$ be the unit $\ell_1$-ball in the space of functions $\lambda:X \to \Rbb$  and $B = \set{y\in \Rbb^m: \|y\|_2 \leq 1}$ be the unit $\ell_2$-ball in $\Rbb^m$.
Define
$$\Phi(y, \lambda)  = \sum_{v\in X} \left(\lambda(v)(\langle u, v\rangle - \langle y, f(v)\rangle) - \Cconst|\lambda(v)| \sqrt{\varepsilon}(\|v\|^2 + 1)\right).$$
We shall prove that there exists $u'\in B$ such that for every $\lambda \in \Lambda$, $\Phi(u', \lambda) \leq 0$.
Observe that this $u'$ will satisfy the statement of the lemma for the following reason.
First, $\|u'\| \leq 1$. Second, let $I_v\in \Lambda$ be the indicator function of $v\in X$; then
%$\lambda \in \Lambda$,
$\Phi(u', I_v) \leq 0$
and $\Phi(u', -I_v) \leq 0$. Therefore,
$|\langle u, v\rangle - \langle u', f(v)\rangle| \leq \Cconst\sqrt{\varepsilon} (\|v\|^2 + 1)$,
as required.

To prove that such $u'$ exists, we show that $\min_{y\in B} \max_{\lambda\in\Lambda} \Phi(y, \lambda) \leq 0$. Note that
$\Lambda$ and $B$ are compact convex sets, $\Phi$ is linear in $y$ and concave in $\lambda$; thus, by the von Neumann minimax theorem~\cite{vonNeumann},
$$\min_{y\in B} \max_{\lambda\in\Lambda} \Phi(y, \lambda)  =  \max_{\lambda\in\Lambda} \min_{y\in B} \Phi(y, \lambda) .$$
Let $\hat\lambda\in \Lambda$ be the $\lambda$ that maximizes the expression on the right. We need to prove that there is $\hat y\in B$ s.t. $\Phi(\hat y, \hat \lambda) \leq 0$.
Define the point $P = \sum_{v\in V} \hat\lambda(v) v$ and
$P' = \sum_{v\in V} \hat\lambda(v) f(v)$. For every $y\in B$, we have
$$\Phi(y, \hat\lambda) = \langle u, P\rangle - \langle y, P'\rangle - \Cconst \sqrt{\varepsilon} \sum_{v\in X} |\hat\lambda(v)| \|v\|^2 - \Cconst\sqrt{\varepsilon} \|\hat \lambda\|_1.$$
Now, $\langle u, P\rangle \leq \|P\|$ since $\|u\| \leq 1$. Let $\hat y = P'/\|P'\| \in B$. We have,
$$\Phi(\hat y, \hat\lambda) \leq \|P\| - \|P'\| - \Cconst \sqrt{\varepsilon} \sum_{v\in X} |\hat\lambda(v)| \|v\|^2 - \Cconst\sqrt{\varepsilon} \|\hat \lambda\|_1.$$
If $\|P\| \leq \|P'\|$ then $\Phi(\hat y, \hat\lambda) \leq 0$ and we are done.
Similarly, if $\|P\| \leq \Cconst \sqrt{\varepsilon} \sum_{v\in X} |\hat\lambda(v)| \|v\|^2$, we are done. We assume below that
$\|P\| > \|P'\|$ and $\|P\| > \Cconst \sqrt{\varepsilon} \sum_{v\in X} |\hat\lambda(v)| \|v\|^2$.
Then,
$$\|P\| - \|P'\| = \frac{\|P\|^2 - \|P'\|^2}{\|P\| + \|P'\|} \leq \frac{\|P\|^2 - \|P'\|^2}{\|P\|} =
\frac{1}{\|P\|} \sum_{v,w \in X} \hat\lambda(v)\hat\lambda(w) (\langle v, w\rangle - \langle f(v), f(w)\rangle).$$
%Since $f$ is non-contracting and $(1+\varepsilon)$-Lipschitz,
Since $f$ satisfies bi-Lipschitz condition (\ref{eq:bi-Lip}) and $\|v-w\|^2 \leq 2(\|v\|^2 + \|w\|^2)$, we have
\begin{align*}
|\langle v, w\rangle - \langle f(v), f(w)\rangle| &= \frac12\left| \|f(v)-f(w)\|^2 - \|f(v)\|^2 - \|f(w)\|^2 -
\|v-w\|^2 + \|v\|^2 + \|w\|^2\right|
\\
&\stackrel{\text{\footnotesize by (\ref{eq:bi-Lip})}}{\leq}
\frac{3\varepsilon}{2} \max(\|v - w\|^2, \|v\|^2 + \|w\|^2)
\leq
3\varepsilon (\|v\|^2 + \|w\|^2).
\end{align*}
Finally, using that $\sum_{v\in V} |\hat\lambda(v)| =  \|\hat \lambda\|_1$ and $\|P\| > \Cconst \sqrt{\varepsilon} \sum_{v\in X} |\hat\lambda(v)| \|v\|^2$, we obtain
$$\|P\| - \|P'\| \leq \frac{3\varepsilon}{\|P\|}\sum_{v,w\in X} |\hat\lambda(v)\hat\lambda(w)| (\|v\|^2 + \|w\|^2\bigr) =
\frac{6\varepsilon \|\hat \lambda\|_1}{\|P\|}\sum_{v\in X} |\hat\lambda(v)| \|v\|^2 \leq \frac{6\varepsilon \|\hat \lambda\|_1}{3\sqrt\varepsilon} \leq 2\sqrt{\varepsilon}.$$
Therefore, $\Phi(\hat y, \hat\lambda) < 0$.
\end{proof}

Now we proceed with the proof of Theorem~\ref{thm:one-point-ext}. Let $u'\in \Rbb^m$ as in Lemma~\ref{lem:one-point-ext-main} and $w' = \sqrt{1 - \|u'\|^2} e_{m+1}$ (where $e_{m+1}$
is a standard basis vector for $\Rbb^{m+1}$). Note that $w'$ is orthogonal to all vectors $f(v) \in \Rbb^m$.
Extend $f$ to $f'$ by letting $f'(u) = u' + w'$. Then,
$\|f'(u)\|^2 = \|u'\|^2 + \|w'\|^2 = 1$.
For every $v\in X$, we have
\begin{align}
\|f'(v) - f'(u)\|^2 &= \|w'\|^2 + \|f(v) - u'\|^2 = (\|w'|^2 + \|u'\|^2) + \|f(v)\|^2 - 2\langle f(v), u'\rangle \\
&= 1 + \|f(v)\|^2 - 2\langle f(v), u'\rangle,\label{eq:diff-1}\\
\|v - u\|^2 &=  1+ \|v\|^2 - 2\langle v, u\rangle.\label{eq:diff-2}
\end{align}
From bounds $\|v - u\|^2 \geq 1$ and $\|v - u\|^2 \geq (\|v\| - 1)^2$, it easily follows that
$\|v - u\|^2 \geq (\|v\|^2 + 1)/5$.
By (\ref{eq:diff-1}), (\ref{eq:diff-2}), and the bound on $|\langle f(v), u'\rangle - \langle v, u\rangle|$ from Lemma~\ref{lem:one-point-ext-main}, we have
$$\bigl|\|f'(v) - f'(u)\|^2 - \|v - u\|^2\bigr| \leq 3\varepsilon \|v\|^2 +6\sqrt\varepsilon (\|v^2\| + 1) \leq
9 \sqrt\varepsilon (\|v^2\| + 1) \leq
45 \sqrt\varepsilon \|v - u\|^2 .$$
This implies that $f'$ has distortion $1+ O(\sqrt\varepsilon)$.
\end{proof}

\subsection{Lower bound}
In this section, we show that the bound in Theorem~\ref{thm:one-point-ext} is tight (up to a constant factor in the $O$-notation) -- extending a map with distortion $1+\eps$ by one point might require blowing up the distortion to $1+\Omega(\sqrt \eps)$,
even when $n = m = 1$ (the extension $f'$ may use extra dimensions).

The construction is as follows. Consider points: $A = 0$, $B = \eps$, $B' = -\eps$, and $C = 1$. Let $X = \{A,B,C\}$.
Consider map $f:X\to \Rbb$ that maps $A$, $B$, $C$ to points $A$, $B'$, $C$, respectively.
Clearly $f$ has distortion $\frac{C-B'}{C-B} = \frac{1+\eps}{1-\eps} \leq 1+3\eps$ for $\eps \leq 1/3$. Our goal is to extend $f$ to the fourth point $D=\sqrt \eps$. Note that we can assume that the extension uses at most one additional dimension.

\begin{claim}
Any outer extension of the map $f$ to the point $D$ has distortion at least $(1+\sqrt \eps/2)$.
\end{claim}
\begin{proof}
Let $f(D) = (x,y)\in\Rbb^2$, and suppose that the distortion is less than $(1+{\sqrt \eps}/2)$. Then we must have
\begin{itemize}
\item $\norm{f(D)-f(A)} \geq (1-{\sqrt \eps}/2) \norm{D-A}$, so $x^2+y^2 \geq (\sqrt \eps - \eps/2)^2$.
\item
$\norm{f(D)-f(B)} \leq (1+{\sqrt \eps}/2) \norm{D-B}$, so $(x+\eps)^2+y^2 \leq \bigl((1+{\sqrt \eps}/2)(\sqrt \eps - \eps)\bigr)^2 \leq (\sqrt \eps - \eps/2)^2$.
\end{itemize}
We get that
$x^2+y^2 \geq (x+\eps)^2+y^2$. Thus, $x\leq - \eps/2$.
Then
$\displaystyle
\frac{\norm{f(D)-f(C)}}{\norm{D-C}} \geq \frac{1}{1-\sqrt \eps} \geq 1+\sqrt \eps,
$
which is a contradiction.
\end{proof} 
\section{Applications -- prioritized and terminal dimension reductions}\label{sec:applications} %\subsubsection{Prioritized dimension reduction}\label{sec:prjl}
In this section, we prove Theorems~\ref{thm:prioritized-embedding} and~\ref{thm:embedding-with-terminals}.

\begin{proof}[Proof of Theorem~\ref{thm:prioritized-embedding}]
First, we construct a $(c_1\log\log j, c_2\log^{3 + \varepsilon} j)$-prioritized dimension reduction. Denote $C = 3 + \eps$.
We define an increasing family of $T = \lceil\log_C\log_2\log_2 N\rceil$ subsets $S_0, S_1,  \dots, S_T$ of $X \subset \Rbb^d$: $S_i$ consists of the first $\min(2^{2^{C^i}},N)$ points according to the priority ranking $\pi$.

For each set $S_i$, we construct an embedding $f_i:S_i \to \Rbb^{d_i}$
with distortion at most $C^i$ for $d_i = O(\log |S_i|)$ in such a way that each $f_i$ is an outer extension of $f_{i-1}$.
We start with $S_0$ -- we let $f_0$ be an isometric embedding of $S_0$ (which consists of $4$ points) into $\Rbb^3$.
Then we iteratively construct mapping $f_{i}$.
At iteration $i$, we take map $f_{i-1}$ and extend it to map $f_i$ as follows.
Using Theorem~\ref{lem:bi-lip-ext}, we find an outer-bi-Lipschitz extension $h:S\to \Rbb^{d'}$ of $f_{i-1}$ to $S_i$.
The extension $h$ is not yet what we want:
\begin{itemize}
\item while, by Theorem~\ref{lem:bi-lip-ext}, its distortion is at most $3\cdot (3+\varepsilon)^{i-1}$, which is less than $C^i$ (the desired upper bound on the distortion),
\item the dimension $d'$ is possibly greater than $\Omega(\log |S_i|)$.
\end{itemize}
To reduce the dimension, we write $h(x) = h_1(x) \oplus h_2(x) \in \Rbb^{d_{i-1}} \oplus \Rbb^{d'- d_{i-1}}$,
here $h_1(x)$ is the vector consisting of the first $d_{i-1}$ coordinates of $h(x)$ and $h_2(x)$ is the vector
consisting of the remaining coordinates of $h(x)$. Since $h$ is an extension of $f_{i-1}$, we have
$h_1(x) = f_{i-1}(x)$ and $h_2(x) = 0$ for $x\in S_{i-1}$. Now, we use the Johnson--Lindenstrauss lemma
to find a dimension reduction $g$ from $h_2(S_i)$ to $\Rbb^{d''}$ with distortion at most $1 + \eps/3$, where
$d'' = c_{JL}\log |S_i|/\varepsilon^2$ for some absolute constant $c_{JL}$. We assume that $g(0) = 0$ (if necessary, we redefine $g$
as $g'(x) = g(x) - g(0)$). Finally, we let $f_i = (id \oplus g)\circ h$;
in other words, $f_i(x) =  h_1(x) \oplus g(h_2(x))$.
\iffalse
\begin{itemize}
\item compute $h(x)$;
\item keep the first $d_{i-1}$ coordinates of $h(x)$ as is;
\item apply the dimension reduction $g$ to the remaining $d' - d_{i-1}$ coordinates;
\item obtain $f_i(x)$ by concatenating the first $d_{i-1}$ coordinates of $h(x)$ and the coordinates of $g(x)$.
\end{itemize}
\fi

Note that $f_i(x)$ is an outer extension of $f_{i-1}$, since $f_i(x) = h_1(x) \oplus g(h_2(x)) = f_{i-1}(x) \oplus g(0) = f_{i-1}(x)$ for
$x\in S_{i-1}$.
The distortion of $id \oplus g$ is at most the distortion of $g$, which is at most $1 + \varepsilon/3$;
therefore, the distortion of $f_i$ is at most $(1 + \varepsilon/3)\times  3\cdot (3+ \varepsilon)^{i-1} = C^i$.
We bound the dimension
$$d_i =  d_{i-1} + d'' = 4 + \sum_{t=1}^{i} c_{JL}\log |S_t|/\varepsilon^2 \leq 4 + \sum_{t=1}^{i-1} c_{JL} 2^{C^t}/\varepsilon^2
+ c_{JL} \log |S_i|/\varepsilon^2 = O(\log |S_i|).$$
The constant in the big-$O$ notation is proportional to $1/\varepsilon^2$.

Finally, let $f = f_T$. We verify that $f$ is $(c_1\log\log j, c_2\log^{3 + \varepsilon} j)$-prioritized dimension reduction.
Fix some $j\in \set{1,\dots, N}$. Let $S_i$ be the smallest of the sets $S_0, \dots, S_T$ that contains $\pi(j)$;
i.e.,  $i = \lceil\log_C\log_2\log_2 j\rceil$ if $j > 4$, and $i = 0$ otherwise.
Then $f$ restricted to $\pi(1), \dots, \pi(j)$ coincides with $f_{i}$. The distortion of
$f_i$ is at most (for $j\geq 4$)
$$C^i \leq C^{1 + \log_C\log_2\log_2 j} \leq C \log_2\log_2 j = (3+\varepsilon) \log_2\log_2 j.$$
Further, $f(\pi(j)) = f_i(\pi(j)) \in \Rbb^{d_{i}}$. Hence, in the vector $f(\pi(j))$ all but the first $d_i$ coordinates
are equal to $0$; we upper bound $d_i$ as follows (for $j\geq 4$):
$d_i \leq O(\log |S_i|) \leq O(2^{C^i}) \leq O\left((2^{C^{i-1}})^C\right) \leq O(\log j)^C$,
as required. Note that the image of $X$ under $f$ lies in space $\Rbb^{d_T}$ of dimension $d_T = O(\log |S_T|) = O(\log N).$

By setting the parameters differently, we can obtain different trade-offs between the distortion and dimension.
Fix a parameter $k \in {\mathbb N}$, $1 < k < \log\log\log N$. Let $T = k$ and $S_i$ be the set consisting of the first $2^{\log_2^{i/k} N}$ points in $X$, according to the priority ordering $\pi$. Construct maps $f_i$ as described above.
The distortion of $f$ is at most
$C^T = (3+\varepsilon)^k$. The vector $f(\pi(j))$ lies in the space $\Rbb^{d_i}$, where $i = \lceil k \frac{\log_2\log_2 j}{\log_2\log_2 N}\rceil$
and
\begin{align*}
d_i &\leq \frac{c_{JL}}{\varepsilon^2} \sum_{t = 0}^i \log |S_t| =
%O\bigl(\sum_{t = 0}^i \log 2^{\log_2^{t/k} N}\bigr)=
O\bigl(\sum_{t = 0}^i \log_2^{t/k} N\bigr) \leq
O(\log_2^{i/k} N) = O\bigl(\underbrace{\log_2^{(i-1)/k} N}_{\text{less than } \log_2 j} \cdot \log_2^{1/k} N\bigr)\\
&\leq O(\log_2 j\, \log_2^{1/k} N).
\end{align*}
We can compute map $f$ in polynomial time, since, at each iteration, we can compute the outer extension $h$ and dimension reduction $g$ in polynomial time.
\end{proof}

Now we prove Theorem~\ref{thm:embedding-with-terminals}.
\begin{proof}[Proof of Theorem~\ref{thm:embedding-with-terminals}]
First we apply the Johnson--Lindenstrauss lemma to $X$ with $\eps' = \varepsilon^2$. We get an embedding
$g:X \to \Rbb^{d'}$ with the distortion at most $1 + \varepsilon^2$ and $d' = O(\log N / \varepsilon^4)$; we rescale it so that
$\lambda \|x - y\| \leq \|g(x) - g(y)\| \leq \lambda(1 + \varepsilon^2)\|x - y\|$,
where $\lambda  = 1 + c\varepsilon$ (we will specify $c$ later).

For every point $p \in \Rbb^d$, we extend $g$ to a map $g_p:X \cup \set{p} \to \Rbb^{d'+1}$ using Theorem~\ref{thm:one-point-ext};
for $p \in X$, $g_p = g$. The distortion of $g_p$ is $1 + O(\sqrt{\varepsilon^2}) = 1 + O(\eps)$.
Finally, we let $f(p) = g_p(p)$.
The image of $f$ lies in $\Rbb^{d'+1}$, as required.
For every $x\in X$ and $p \in \Rbb^d$, we have $g_x(x) = g(x) = g_p(x)$ and
$$\|f(p) - f(x)\| = \|g_p(p) - g_x(x)\| = \|g_p(p) - g_p(x)\| \in [(1+c\eps)(1-O(\eps)) \|p - x\|, (1+c\eps)(1+O(\eps)) \|p - x\|].$$
We choose $c$ so that the $(1+c\eps)(1-O(\eps))$ term is $1$; then $(1+c\eps)(1+O(\eps)) = 1 + O(\eps)$.

Note that we can compute $f(x)$ in polynomial time, since we can compute each map $g_p$ in polynomial time.
\end{proof} 

%............................

\section{Overview of the extension result for maps from \texorpdfstring{$\Rbb$}{R} to \texorpdfstring{$\Rbb$}{R}.} \label{sec:line-overview}
%\definecolor{gr}{RGB}{127,127,127}
%\newcommand{\yinsert}[1]{{\color{blue} #1}}
%\newcommand{\ychange}[2]{{\color{blue} #1} {\color{gr} (#2)}}
%\newcommand{\ychangenote}[3]{{\color{blue} #1} {\color{gr} (#2)}{\ynote{#3}} }

In this section, we consider the case of map $f \colon X \to \Rbb$ with distortion $(1+\eps)$, where $X \subset \Rbb$.
We show that such a map is very structured, which allows us to extend it to $\tilde f\colon \Re\to \Re^2$ with the 	 distortion
$1+O(1/\log^2(1/\eps))$. Here we provide an informal overview to illustrate the main steps.

First, suppose that $X$ consists of three points $0,\eps,1$ that $f$ maps to $0,-\eps,1$, respectively.
It turns out that this simple case is in fact very important. We extend $f$ to the whole segment $[0;1]$ as follows\footnote{Extending $f$ to the whole $\Rbb$ requires a bit more work.}.
For $0\leq x \leq \eps$, we map $x$ to $(-x,0)$, and for $\eps \leq x \leq 1$, we map $x$ to point $g(x) = (r(x),\varphi(x))$ in polar coordinates, where the radius is $r(x)=x$ and the angle is $\varphi(x) = \frac{\pi\ln(1/x)}{\ln(1/\eps)}$, see Figure~\ref{spiral_fig_1}, page~\pageref{spiral_fig_1}. First, the map is continuous (i.e., $g(\eps)=-\eps$ and $g(1) = 1$). Second, for every $x$, $\|g(x)\| = |x|$, which implies that $g$ is non-contractive. We refer to this map as the ``spiral''. We prove that its distortion is $1+O(1/\ln^2(1/\eps))$, and in fact this is the optimal distortion one can achieve for this specific choice of $X$ and $f$ (see Section \ref{sec:line-lower} for the proof).

For the general case, we decompose $f$ into ``flips'' and use this decomposition to assemble the extension from the above spirals on various distance scales.

For a set $X$ and map $f$, consider how $f$ changes the relative ordering of points $X$; denote the corresponding permutation by $\pi_f \in S_{|X|}$.
For instance, if $X = \{x_1, x_2, x_3\}$, where $x_1 < x_2 < x_3$, and $f(x_1) < f(x_3) < f(x_2)$,
we set $\pi_f = (1\ 3\ 2)$. We show that a permutation can arise as $\pi_f$ for some $f$ iff it excludes $(3\ 1\ 4\ 2)$ and $(2\ 4\ 1\ 3)$
as a subpermutation. Furthermore, we show that $\pi_f$ can be decomposed into a laminar sequence of \emph{flips}.
We start with the identity permutation, and then iteratively choose a substring and reverse its order (this is one flip).
We do this so that every two flips are either disjoint, or the later is strictly contained in the earlier one.
For example, if $\pi_f = (3\ 1\ 2\ 4\ 6\ 5)$, then the decomposition is as follows: $(1\ 2\ 3\ 4\ 5\ 6)$, $({\color{red}3\ 2\ 1}\ 4\ 5\ 6)$,
$(3\ {\color{red}1\ 2}\ 4\ 5\ 6)$, $(3\ 1\ 2\ 4\ {\color{red}6\ 5})$.

We use this decomposition to build the desired extension. For each flip, we add two spirals. We show that the points that participate in a given flip are well-separated from others. For example if the permutation is $(1\ 3\ 2)$, then the distance between $2$ and $3$ should be much smaller by a factor of $\eps$) than the distance from $1$ to either of them -- both in the domain and in the image.
We show that this separation is sufficient for these spirals not to interfere much with each
other, and the bound of $1 + O\left(1 / \log^2(1 / \eps)\right)$ on the distortion holds for the overall construction.
See Figure~\ref{spiral_fig_1}, page~\pageref{spiral_fig_1}, for the construction for the case $\pi_f = (3\ 1\ 2\ 4\ 6\ 5)$.

%\paragraph{{\color{red}Overview of the lower bound.}}...

{
\small
\bibliographystyle{alpha}
\bibliography{bibfile}
}
\appendix
\section{Outer extension of a map from \texorpdfstring{$\Rbb$}{R} to \texorpdfstring{$\Rbb$}{R}}\label{sec:line-case}
\subsection{Extension to the whole line}\label{sec:line}
In this section we prove the following theorem.

\begin{theorem}[Theorem~\ref{thm:line-embedding}]
\label{line_main}
  Let $X \subset \Rbb$ be an arbitrary set.
  Suppose that $f \colon X \to \Rbb$ is a map such that for every $x_1, x_2 \in X$, we have:
  \begin{equation}
  \label{near_iso_line}
  |f(x_1) - f(x_2)| \in (1 \pm \eps) \cdot |x_1 - x_2|.
  \end{equation}
  Then there exists a map $h \colon \Rbb \to \Rbb^2$ such that:
  \begin{itemize}
  \item For every $x \in X$, we have $h(x) = (f(x), 0)$;
  \item For every $u, v \in \Rbb$, we have
    $$
    \|h(u) - h(v)\| \in \left(1 \pm O\left(\frac{1}{\log^2(1 / \eps)}\right)\right) \cdot |u - v|.
    $$
  \end{itemize}
\end{theorem}

By a standard compactness argument, it is enough to handle the case of a finite $X$. From now on, we denote
$n = |X|$.

\begin{figure}
\centering
\begin{subfigure}{.45\textwidth}
\centering
\includegraphics{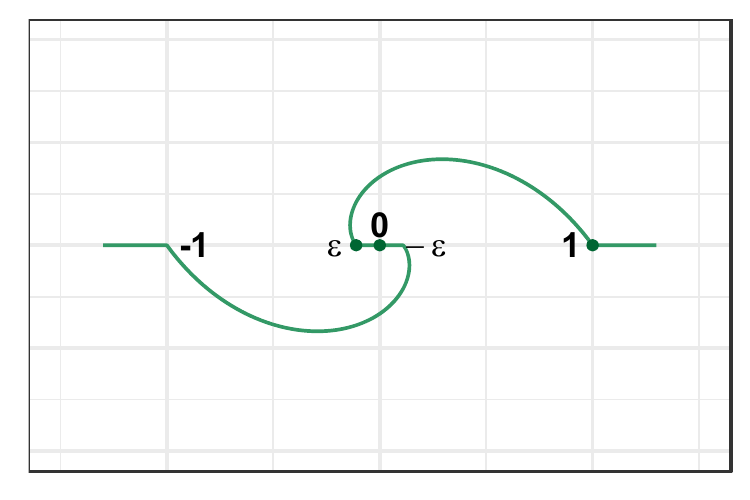}
\end{subfigure}%
\begin{subfigure}{.45\textwidth}
\centering
\includegraphics{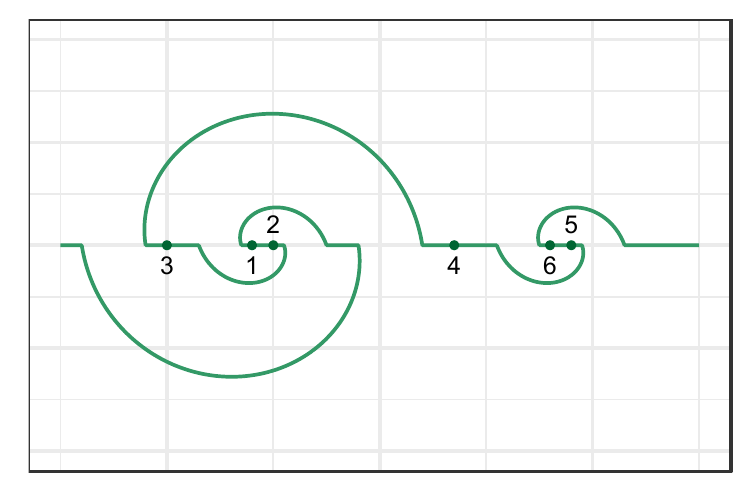}
\end{subfigure}
\caption{Left: one possible extension for the map $0 \mapsto 0$, $\varepsilon \mapsto -\varepsilon$, $1 \mapsto 1$. It has distortion
$1 + O(1 / \log^2(1 / \eps))$, which is tight for this example. Right: an extension built from the spirals recursively for the map $f$
with $\pi_f = (3\ 1\ 2\ 4\ 6\ 5)$. The picture is intentionally out of proportion.}
\label{spiral_fig_1}
\end{figure}

\subsubsection{Characterizing near-isometric maps} To prove the main theorem, we will first prove the necessary conditions $f$ needs to satisfy in order to be a near-isometric mapping. In the rest, we will denote the initial point set by $X = \{x_1, x_2, \ldots, x_n\}$ and without loss of generality we may assume that $x_1 < x_2 < \ldots < x_n$.
Let $\pi_f \in S_n$ be the permutation defined by our mapping $f$ such that $f(x_{\pi_f(1)}) < f(x_{\pi_f(2)}) < \ldots < f(x_{\pi_f(n)})$. The following lemma characterizes the properties of $\pi_f$.

\begin{definition}[Sub-permutation]
Given a permutation $\sigma$ of $[k]$, and a permutation $\pi$ of $[n]$, where $n\geq k$, we say that $\pi$ contains $\sigma$ as a sub-permutation iff there exists $i_1<\cdots<i_k \in [n]$ such that for any $j,j'\in [k]$, if $\sigma(j) < \sigma(j')$, then $\pi(i_j)<\pi(i_{j'})$.
\end{definition}

\begin{lemma}
\label{bad_perm_lemma}
If $\eps > 0$ is sufficiently small, then $\pi_f$ does not have $(3\ 1\ 4\ 2)$ or $(2\ 4\ 1\ 3)$ as sub-permutations.
\end{lemma}
\begin{proof}
Let us prove the statement for $(3\ 1\ 4\ 2)$, the proof for $(2\ 4\ 1\ 3)$ is the same.
Assume the contrary. Then, there exists $1 \leq i < j < k < l \leq n$ such that
\begin{equation}
\label{bad_perm}
f(x_k) < f(x_i) < f(x_l) < f(x_j).
\end{equation}
Denote $\Delta = x_l - x_i > 0$.
Then,
\begin{align*}
\Delta = x_l - x_i & \geq (x_l - x_k) + (x_j - x_i) \\&\geq (1 - O(\eps)) \cdot ((f(x_l) - f(x_k)) + (f(x_j) - f(x_i)))
\\& \geq (2 - O(\eps)) \cdot (f(x_l) - f(x_i))
\\& \geq (2 - O(\eps)) \cdot (x_l - x_i)
\\&= (2 - O(\eps)) \cdot \Delta,
\end{align*}
where the first step follows from $x_i < x_j < x_k < x_l$ (which in turn follows from $i < j < k < l$),
the second step follows from $f$ having distortion $(1 + \eps)$ and from~(\ref{bad_perm}),
and the fourth step again follows from $f$ being a near-isometry. Thus, if $\eps > 0$ is sufficiently small, we get a contradiction.
\end{proof}

\subsubsection{Permutation decomposition}
\begin{lemma}
\label{perm_decomposition}
If $\eps > 0$ is sufficiently small, then $\pi_f$ can be decomposed as follows.
We start with $\pi_0$ which is the identity permutation.
Then, we perform $T \geq 0$ \emph{flips} as follows.
Each flip $1 \leq t \leq T$ is defined by two numbers $1 \leq a_t < b_t \leq n$, naturally defining a segment in the permutation.
We obtain $\pi_t$ from $\pi_{t-1}$ as follows.
$$
\pi_t(k) = \begin{cases}
\pi_{t-1}(a_t + b_t - k),&\mbox{if $a_t \leq k \leq b_t$,}\\
\pi_{t-1}(k),&\mbox{otherwise.}\\
\end{cases}
$$
In words, we obtain $\pi_t$ from $\pi_{t-1}$ be reversing the segment $[a_t, b_t]$.
Moreover, the segments form a laminar family: for every $1 \leq t_1 < t_2 \leq T$ the segments $[a_{t_1}, b_{t_1}]$ and $[a_{t_2}, b_{t_2}]$ are either disjoint
or $[a_{t_1}, b_{t_1}] \supset [a_{t_2}, b_{t_2}]$.
The permutation $\pi_f$ is equal to the final permutation $\pi_T$.
\end{lemma}
\begin{proof}
The proof is by induction over $n$. If $n = 1$, the statement is trivial.
Denote $1 \leq u \leq n$ such that $\pi_f(u) = 1$ (the position where $1$ is mapped to), and $1 \leq v \leq n$ such that $\pi_f(v) = n$ (the position where $n$ is mapped to).
Suppose that $u < v$.
If $u = 1$, then the statement follows from using the induction assumption on $\pi_f$ without the first element.
Assume that $u > 1$. Then, define $A = \{\pi_f(j) \mid j \leq u\}$, to be the set of numbers that are mapped to the left of $1$. Let $z < u$ be such that $\pi_f(z) = \max A$, i.e., the maximum number mapped to the left of $1$. Define $w = \min\{k \mid \pi_f(k) > \max A\}$.
Clearly, $w \leq v$. We claim that the sequence $(\pi_f(1)\ \pi_f(2)\ \ldots\ \pi_f(w - 1))$
is a permutation of the numbers from $1$ to $z$.
Assume not. Then, there exists $w' > w$ such that $\pi_f(w') < z$.
 Then, considering positions $z$, $u$, $w$, and $w'$,
we obtain a sub-permutation $(3\ 1\ 4\ 2)$, which can not be the case by Lemma~\ref{bad_perm_lemma}.
Now we can apply the inductive assumption on the first $w - 1$ numbers, and on the last $n - w + 1$ numbers,
and merge the resulting sequences of flips.
If $u > v$, then we add a flip with $a = 1$ and $b = n$ and reduce to the case, when $u < v$.
\end{proof}
\iffalse
\begin{proof}
The proof is by induction over $n$. If $n = 1$, the statement is trivial.
Denote $1 \leq u \leq n$ such that $\pi_f(u) = 1$, and $1 \leq v \leq n$ such that $\pi_f(v) = n$.
Suppose that $u < v$.
If $u = 1$, then the statement follows from using the induction assumption on $\pi_f$ without the first element.
Assume that $u > 1$. Then, define $A = \{\pi_f(j) \mid j \leq u\}$. Define $w = \min\{k \mid \pi_f(k) > \max A\}$.
Clearly, $w \leq l$. We claim that the sequence $(\pi_f(1)\ \pi_f(2)\ \ldots\ \pi_f(w - 1))$
is a permutation of the numbers from $1$ to $\max A$.
Assume not. Then, there exists $w' > w$ such that $\pi_f(w) < \max A$.
Let $w'' < u$ be the position such that $\pi_f(w'') = \max A$. Then, considering positions $w''$, $u$, $w$, and $w'$,
we obtain a subpermutation $(3\ 1\ 4\ 2)$, which can not be the case by Lemma~\ref{bad_perm_lemma}.
Now we can apply the inductive assumption on the first $w - 1$ numbers, and on the last $n - w + 1$ numbers,
and merge the resulting sequences of flips.
If $u > v$, then we add a flip with $a = 1$ and $b = n$ and reduce to the case, when $u < v$.
\end{proof}
\fi

It is not hard to show that the above condition is also a sufficient condition, but we will not need it in our construction.

\subsubsection{Well-separateness and the portals}
First, for each flip $1 \leq t \leq T$, we define the set of points $F_t$ that are affected by it,
the set of points to the left of $F_t$, denoted $L_t$, and the points to the right, $R_t$.
Formally, we have the following.
\begin{definition} For an iteration $1 \leq t \leq T$, we define
\begin{itemize}
\item $L_t = \{\pi_t(1), \pi_t(2), \ldots, \pi_t(a_t - 1)\}$;
\item $F_t = \{\pi_t(a_t), \pi_t(a_t + 1), \ldots, \pi_t(b_t)\}$;
\item $R_t = \{\pi_t(b_t+1), \pi_t(b_t + 2), \ldots, \pi_t(n)\}$.
\end{itemize}
\end{definition}

\begin{lemma}
$F_t$ is the set of $|F_t| = b_t - a_t + 1$ consecutive integers. Moreover,
the sequence $\pi_t(a_t), \pi_t(a_t + 1), \ldots, \pi_t(b_t)$ is either increasing or decreasing.
\end{lemma}
\begin{proof}
Follows trivially from Lemma~\ref{perm_decomposition}.
\end{proof}

\begin{definition}
For an iteration $t\leq T$, we define $u_t = \pi_{t-1}(a_t)$ and $v_t = \pi_{t-1}(b_t)$.
We also define $\Delta_t = x_{v_t} - x_{u_t}$. It can be either positive or negative.
\end{definition}

The quantity $\Delta_t$ can be seen as the \emph{signed} diameter of the flipped points.
The following lemma is a key to the overall analysis. We show that the flipped points $F_t$ are very well-separated
from the remainder: by the amount $\Omega(|\Delta_t| / \eps)$.

\begin{lemma}
\label{lem_separ}
For every $k \in F_t$, and every $p \in L_t \cup R_t$, we have $|x_k - x_p| \geq \Omega\left(\frac{|\Delta_t|}{\eps}\right)$.
\end{lemma}
\begin{proof}
  Wlog, we can assume that $t$ is the first flip that separates $p$ and $k$ and for which $k \in F_t$, but $p \notin F_t$.
  Indeed, if $\widetilde{t} < t$ is the first such flip,
then $|\Delta_{\widetilde{t}}| > |\Delta_t|$, and the required statement follows from that about $\widetilde{t}$.
Suppose that $p \in L_t$, the case $p \in R_t$ is similar. Then, we have $f(x_p) < f(x_{v_t}) < f(x_{u_t})$
(here we use crucially the fact that $t$ is the first flip that separates $p$ and $k$). Indeed, $t$
is the last flip, which affects the relative order of $f(x_p)$, $f(x_{v_t})$ and $f(x_{u_t})$,
since the flips that are not disjoint are nested.
At the same time, either $x_p < x_{u_t} \leq x_k \leq x_{v_t}$ or $x_p > x_{u_t} \geq x_k \geq x_{v_t}$.
Let us show how to handle the first case, the second case is similar. Let us denote $s = x_{u_t} - x_p$.
See Figure~\ref{fig_separation} for the clarification.
Then,
\begin{align*}
s(1 + \eps) & \geq f(x_{u_t}) - f(x_p) \\&= (f(x_{u_t}) - f(x_{v_t})) + (f(x_{v_t}) - f(x_p)) \\&\geq
(1 - \eps)(x_{v_t} - x_{u_t}) + (1 - \eps)(x_{v_t} - x_p) \\&= (1 - \eps) (s + 2 \Delta_t).
\end{align*}
Thus, $\Delta_t = O(\eps\cdot s)$. Finally, $|x_k - x_p| \geq s = \Omega(\Delta_t / \eps)$.
\end{proof}

\begin{figure}
  \centering
  \includegraphics{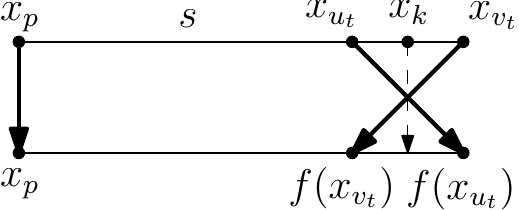}
  \caption{Illustration to the proof of Lemma~\ref{lem_separ}}
  \label{fig_separation}
\end{figure}

\begin{definition}[Portals] For every $1 \leq t \leq T$, we define \emph{portals} as follows (see Figure~\ref{fig_portals}). We set:
\begin{itemize}
\item $\alpha_t = x_{u_t} - \frac{\Delta_t}{\eps^{2/3}}$;
$\beta_t = x_{u_t} - \frac{\Delta_t}{\eps^{1/3}}$;
$\gamma_t = x_{v_t} + \frac{\Delta_t}{\eps^{1/3}}$;
$\delta_t = x_{v_t} + \frac{\Delta_t}{\eps^{2/3}}$;
\item $\alpha'_t = f(x_{v_t}) - \frac{\Delta_t}{\eps^{2/3}}$;
$\beta'_t = f(x_{v_t}) - \frac{\Delta_t}{\eps^{1/3}}$;
$\gamma'_t = f(x_{u_t}) + \frac{\Delta_t}{\eps^{1/3}}$;
$\delta'_t = f(x_{u_t}) + \frac{\Delta_t}{\eps^{2/3}}$.
\end{itemize}
\end{definition}
We will use the portals in our construction to make sure that the spirals at different levels do not interfere with each other.

\begin{figure}
  \centering
  \includegraphics{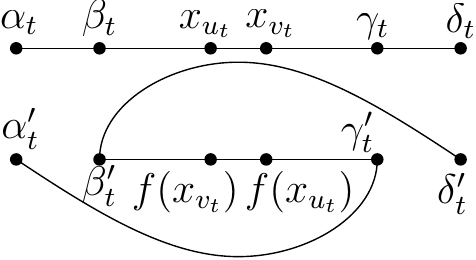}
  \caption{Portals. Note that the scales of the relative distances are not correct.}
  \label{fig_portals}
\end{figure}

\subsubsection{Construction of the final map}
Now we are ready to define the final map $h \colon \Rbb \to \Rbb^2$.
First, for every $1 \leq k \leq n$, we set $h(x_k) = (f(x_k), 0)$.
Second, for every $1 \leq t \leq T$, we define $h$ between $\alpha_t$ and $\beta_t$
and between $\gamma_t$ and $\delta_t$ according to the Corollary~\ref{cor:spirals} (note that we only take the part of the map which corresponds to these two intervals, see Figure~\ref{fig_portals} for the illustration).
In particular, $h(\alpha_t) = (\alpha'_t, 0)$, $h(\beta_t) = (\gamma'_t, 0)$, $h(\gamma_t) = (\beta'_t, 0)$ and $h(\delta_t) = (\delta'_t, 0)$.
After we are done with constructing the spirals for all iterations $t$, on the remaining bounded intervals on the real line, we define $h$ to be linear and consistent with the values at the endpoints.
For the two unbounded intervals, we define the map to be appropriate shifts.

Let us now show that for every $x, y \in \Rbb$, we have:
$$
\|h(x) - h(y)\| \in \left(1 \pm O\left(\frac{1}{\log^2(1 / \eps)}\right)\right) \cdot |x - y|.
$$

For a point $t \in \Rbb$, there are two cases: either it is mapped using the map $g$
from Corollary~\ref{cor:spirals}, or it is mapped using a linear extension.
In the former case, we say that $t$ is of ``type A'', while in the latter case it is said to be of  ``type B''. Note that the type A points are mapped on a spiral curve in $\Re^2$, and the type B points are mapped on a segment in $\Re$.

\begin{claim}
\label{ext_f_iso}
If we extend the original map $f$ to the portals (such that $\alpha_t \mapsto \alpha_t'$,
$\beta_t \mapsto \gamma_t'$, $\gamma_t \mapsto \beta_t'$ and $\delta_t \mapsto \delta_t'$),
then the resulting map is a $(1 \pm O(\eps^{1/3}))$-isometry.
\end{claim}
\begin{proof}
    It is immediate to check that the worst case is achieved when we consider distances between portals
    $\alpha_t$ and $\beta_t$ or $\gamma_t$ and $\delta_t$. In this case,
    the distortion is $1 + \Theta(\eps^{1/3})$ (this follows from the definition of the portals).
\end{proof}

\begin{claim}
\label{gradient_type_b}
    If $t \in \Rbb$ is type B, and $h$ is smooth at $t$, then $\|\nabla h(t)\|_2 = 1 \pm O(\eps^{1/3})$.
\end{claim}
\begin{proof}
    This is a direct corollary of Claim~\ref{ext_f_iso}.
\end{proof}

\begin{claim}
  \label{type_b_type_b}
    If both $x, y \in \Rbb$ are type B, then
    $$
    \|h(x) - h(y)\| \in \left(1 \pm O\left(\eps^{1/3}\right)\right) \cdot |x - y|.
    $$
\end{claim}
\begin{proof}
  If $x = y$, then there is nothing to prove.
  If $x \ne y$ by a small perturbation we can assume wlog that $h$ is smooth in both $x$ and $y$.
  By Claim~\ref{gradient_type_b}, $\|\nabla h(x)\|_2, \|\nabla h(x)\|_2 \in 1 \pm O(\eps^{1/3})$.
  If the signs of $(\nabla h(x))_1$ and $(\nabla h(x))_2$ are the same, then the claim follows
  from Claim~\ref{ext_f_iso} and Claim~\ref{gradient_type_b}.

  Now consider the case of the different signs of the derivatives.
  Then consider an extension of $f$ to the portals as stated in Claim~\ref{ext_f_iso}.
  Abusing notation, let us denote this map $f$ as well.
  Since the extended map has distortion $1 \pm O(\eps^{1/3})$, we decompose it as per
  Lemma~\ref{perm_decomposition}, and we get that Lemma~\ref{lem_separ} holds.

  Let us denote $p_x < x < q_x$ the portals of elements which are closest to $x$,
  similarly, we denote $p_y < y < q_y$. Wlog, $q_x \leq p_y$.
  If a decomposition for $f$ has a flip containing $p_y$ and $q_y$, but not $p_x$ and $q_x$,
  then $p_y - q_x \geq \Omega\left(\frac{q_y - p_y}{\eps^{1/3}}\right)$.
  Similarly, if there is a flip containing $p_x$ and $q_x$, but not $p_x$ and $q_x$,
  then $p_y - q_x \geq \Omega\left(\frac{q_x - p_x}{\eps^{1/3}}\right)$.
  Note that if neither of these two cases hold, then their gradients could not have different signs.
  Combining these observations with Claim~\ref{ext_f_iso} and Claim~\ref{gradient_type_b},
  we get the required result.
\end{proof}

\begin{claim}
    If both $x, y \in \Rbb$ are type A, then
    $$
    \|h(x) - h(y)\| \in \left(1 \pm O\left(\frac{1}{\log^2(1 / \eps)}\right)\right) \cdot |x - y|.
    $$
\end{claim}
\begin{proof}
Define $t_x$ to be the flip $1 \leq t \leq T$, such that $x$ lies between $\alpha_t$ and $\beta_t$
or $\gamma_t$ and $\delta_t$. We define $t_y$ similarly.

If $t_x = t_y$, then the claim follows from Corollary~\ref{cor:spirals}.

First, suppose that $[a_{t_x}, b_{t_x}]$ and $[a_{t_y}, b_{t_y}]$ are disjoint.
Assume wlog that $|\Delta_{t_x}| \geq |\Delta_{t_y}|$.
Then,
\begin{align*}
\|h(x) - h(y)\| & = \|h(\alpha_{t_x}) - h(\alpha_{t_y})\| \pm O(|\Delta_{t_x}| / \eps^{2/3})\\
              & \in (1 \pm O(\eps^{1/3})) |\alpha_{t_x} - \alpha_{t_y}| \pm O(|\Delta_{t_x}| / \eps^{2/3})\\
              & \in (1 \pm O(\eps^{1/3})) |x - y| \pm O(|\Delta_{t_x}| / \eps^{2/3})\\
              & \in (1 \pm O(\eps^{1/3})) |x - y|,
\end{align*}
where the first step follows from Corollary~\ref{cor:spirals}, the second step
follows from Lemma~\ref{lem_separ}, the third step follows from the definition of the terminals,
and the last step follows from Lemma~\ref{lem_separ}.

Now assume that $[a_{t_x}, b_{t_x}] \supseteq [a_{t_y}, b_{t_y}]$, but $t_x \ne t_y$.
Then, we have $|x - y| \geq \Omega(|\Delta_{t_x}| / \eps^{1/3})$, $|\Delta_{t_x}| = \Omega(|\Delta_{t_y}| / \eps)$ and:
\begin{align*}
\|h(x) - h(y)\| & = |h(x) - h(\alpha_{t_y})\| \pm O(|\Delta_{t_y}| / \eps^{2/3})\\
              & \in \left(1 \pm O\left(\frac{1}{\log^2(1 / \eps)}\right)\right)|x - \alpha_{t_y}| \pm O(|\Delta_{t_y}| / \eps^{2/3})\\
              & \in \left(1 \pm O\left(\frac{1}{\log^2(1 / \eps)}\right)\right)|x - y| \pm O(|\Delta_{t_y}| / \eps^{2/3})\\
              & \in \left(1 \pm O\left(\frac{1}{\log^2(1 / \eps)}\right)\right)|x - y|,
\end{align*}
where the first step is due to the definition of the portals and Corollary~\ref{cor:spirals},
the second step is due to Corollary~\ref{cor:spirals},
the third step is again due to the definition of the portals,
and the last step is due to $|x - y| \geq \Omega(|\Delta_{t_x}| / \eps^{1/3}) \geq \Omega(|\Delta_{t_y}| / \eps^{4/3})$.
\end{proof}

\begin{claim}
  \label{claim_ab}
    If $x \in \Rbb$ is type A and $y \in \Rbb$ is type B, then
    $$
    \|h(x) - h(y)\| \in \left(1 \pm O\left(\frac{1}{\log^2(1 / \eps)}\right)\right) \cdot |x - y|.
    $$
\end{claim}
\begin{proof}
    Denote $1 \leq t_x \leq T$ to be the flip such that $x$ lies within $\alpha_{t_x}$ and $\beta_{t_x}$ or between $\gamma_{t_x}$ and
    $\delta_{t_x}$.
    Wlog, let us assume that $x$ lies between $\alpha_{t_x}$ and $\beta_{t_x}$.
    Then, $y$ can lie between $\beta_{t_x}$ and $\gamma_{t_x}$ or outside of the segment connecting $\alpha_{t_x}$
    and $\delta_{t_x}$. Let us assume the former, and the latter can be handled similarly.
    By Corollary~\ref{cor:spirals}, we have:
    \begin{equation}
    \label{part_1_path}
    \|h(y) - h(x)\| \in O\left(1 \pm O\left(\frac{1}{\log^2(1 / \eps)}\right)\right) \cdot |x - \widetilde{y}|,
    \end{equation}
    where $\widetilde{y}-\beta_t = \gamma_t' - h(y)_1$ (see Figure~\ref{fig_type_ab}).
    By Claim~\ref{type_b_type_b},
    $$
    |\beta_t - \widetilde{y}| = \|h(\beta_t) - h(y)\| \in (1 \pm O(\eps^{1/3})) \cdot |\beta_t - y|.
    $$
    Thus,
    \begin{equation}
    \label{part_2_path}
    |y - \widetilde{y}| \in O(\eps^{1/3}) \cdot |\beta_t - y| \leq O(\eps^{1/3}) \cdot |x - y|.
    \end{equation}
    Combining~(\ref{part_1_path}) and~(\ref{part_2_path}), we are done.
\end{proof}

\begin{figure}
  \centering
  \includegraphics{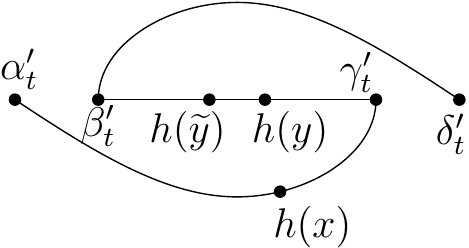}
  \caption{Illustration for the proof of Claim~\ref{claim_ab}.}
  \label{fig_type_ab}
\end{figure}

\subsection{An auxiliary map: the spiral}
\begin{lemma}\label{lem:basic-map}
    Let $\eps > 0$ be a small positive parameter. Let $g\colon \Rbb \to \Rbb^2$ be the map defined as follows. 
   \[
g(t) = \begin{cases}
(t,0), &\mbox{if $|t|>1$,}\\
(-t,0),&\mbox{if $|t|< \eps$}\\
(t\cos \varphi(t),t\sin \varphi(t)),\mbox{where $\varphi(t) = \frac{\pi \ln(1/|t|)}{\ln(1/\eps)}$} &\mbox{otherwise.}\\
\end{cases}
\]
Where the third term can be viewed in the polar coordinates as $\left(r(t) = t, \varphi(t) = \frac{\pi \ln(1/|t|)}{\ln(1/\eps)}\right)$. Then we have the following properties,
\begin{itemize}
    \item{Distortion:} for every $t_1, t_2 \in \Rbb$, one has:
    $$
    \|g(t_1) - g(t_2)\| \in \left(1 + O\left(\frac{1}{\log^2(1 / \eps)}\right)\right) \cdot |t_1 - t_2|;
    $$
    \item{Total movement:} for every $t \in \Rbb$, one has:
    $$
    \|g(t) - (t, 0)\| \leq O\left(1\right).
    $$
\end{itemize}
\end{lemma}
\begin{proof}
First of all note that the function is continuous as $g(\eps) = (\eps\cos \varphi(\eps),\eps\sin\varphi(\eps)) = (-\eps,0)$, $g(-\eps) = (\eps,0)$, $g(1) = (\cos \varphi(1), \sin \varphi(1)) = (1,0)$, and $g(-1) = (-1,0)$. Next we show that the distortion is bounded as desired.
First, we prove that $g$ does not increase the distance by more than a multiplicative factor of $1+O(\frac{1}{\ln^2(1/\eps)})$, and second in Claim \ref{clm:spiral-decrease}, we prove that the distances do not decrease by more than the same factor. These two prove the bound on the distortion as desired.
Finally in Claim \ref{clm:spiral-movement}, we show the total movement property.

\begin{claim}
For $\eps\leq t_1<t_2 \leq 1$, we have $\|g(t_1) - g(t_2)\| \leq \left(1 + O(\frac{1}{\log^2(1 / \eps)})\right) \cdot |t_1 - t_2|$.
\end{claim}
\begin{proof}
The distance between $g(t_1)$ and $g(t_2)$ is at most the length of the curve between them which is given by the following formula
\begin{align*}
&\int_{t=t_1}^{t_2} \sqrt{\left(\frac{d(t\cos\varphi(t))}{dt}\right)^2 + \left(\frac{d(t\sin\varphi(t))}{dt}\right)^2} dt = \\
&\int_{t=t_1}^{t_2} \sqrt{\left(\cos\varphi(t) + \frac{\pi}{\ln(1/\eps)}\sin \varphi(t) \right)^2 + \left(\sin \varphi(t) - \frac{\pi}{\ln(1/\eps)}\cos\varphi(t)\right)^2} dt = \\
&\int_{t=t_1}^{t_2} \sqrt{1 + \left(\frac{\pi}{\ln(1/\eps)}\right)^2	} dt = (t_2-t_1)\sqrt{1 + \left(\frac{\pi}{\ln(1/\eps)}\right)^2	} \leq (t_2-t_1) \left(1+\frac{\pi^2}{2\ln^2(1/\eps)}\right)\\
\end{align*}
\end{proof}
The above claim, together with the fact that the function is symmetric around the origin, and the definition of the function for $|t|\geq 1$ and $|t|\leq \eps$, and triangle inequality, proves that for any $t_1,t_2\in \Re$, the distance between the images, $g(t_1)$ and $g(t_2)$ is increased by at most $\mathcal{D} = 1+O(\frac{1}{\ln^2(1/\eps)})$. Next we prove that the distances do not decrease too much either.
\begin{claim}\label{clm:spiral-decrease}
Given $t_1<t_2$, we have $\|g(t_1) - g(t_2)\| \geq \frac{|t_1 - t_2|}{\mathcal{D}}$.
\end{claim}
\begin{proof}
The claim is trivial if both $|t_1|,|t_2|\geq 1$ or $|t_1|,|t_2|\leq \eps$. Also if $t_2\geq 1$ and $-\eps \leq t_1\leq \eps$, the claim holds as $\frac{t_2-t_1}{g(t_2) - g(t_1)} \leq \frac{t_2 + \eps}{t_2 - \eps} \leq \frac{1+\eps}{1-\eps} \leq 1+3\eps$ for sufficiently small $\eps$. Also if $\eps\leq t_1 < t_2 \leq 1$, then by triangle inequality, $\|g(t_2)-g(t_1)\|\geq \|g(t_2)\|-\|g(t_1)\| = t_2 - t_1$. The remaining cases are discussed bellow or implied by symmetry.
\paragraph{Case 1.} If $\eps \leq t_2\leq  1$ and $-1\leq t_1 \leq -\eps$, by symmetry we can assume that $t_2\geq |t_1|$, and thus suppose that $t_1 = -\alpha  t_2$, where $0 \leq \alpha \leq 1$. 
First, note that if $\alpha \leq 1/\ln^2 (1/\eps)$, then since the distances from the origin to the points remain unchanged, we have that 
$$
\frac{\|g(t_1)-g(t_2)\|}{|t_1-t_2|} \geq \frac{t_2+t_1}{t_2 - t_1} \geq \frac{1-\alpha}{1+\alpha} \geq 1-O(\alpha) \geq 1-O(1/\log^2 (1/\eps))
$$
which proves the claim. Therefore, we can assume that $\alpha \geq 1/\ln^2(1/\eps)$. We should show that $\|g(t_1)-g(t_2)\|/|t_1-t_2| \geq 1/\mathcal{D} \geq 1 - O\left(\frac{1}{\ln^2(1/\eps)}\right)$, or equivalently,  $\|g(t_1)-g(t_2)\|^2/|t_1-t_2|^2 \geq 1 - O(\frac{1}{\ln^2 (1/\eps)})$. 
\begin{align*}
&\frac{\|g(t_1) - g(t_2)\|^2}{|t_1 - t_2|^2} = \frac{t_1^2 + t_2^2 - 2t_1t_2 \cos(\varphi(t_1) - \varphi(t_2))}{(t_1-t_2)^2 } \\
&= 1 + \frac{ 2t_1t_2\left(1-\cos(\varphi(t_1)-\varphi(t_2))\right)}{(t_1 - t_2)^2} =1 - \frac{2t_2^2\alpha (1 - \cos(\varphi(t_1) - \varphi(t_2)))}{t_2^2(1+\alpha)^2}  \\
&= 1 - \frac{2\alpha(1 - \cos(\varphi(t_1) - \varphi(t_2)))}{(1+\alpha)^2} = 1 - O(\alpha(1-\cos(\varphi(t_1)-\varphi(t_2))
\end{align*}
Therefore, we just need to show that $\alpha(1-\cos(\varphi(t_1) - \varphi(t_2))) = O(1/ \ln^2 (1/\eps))$. 
Note that 
$$
\varphi(t_1)-\varphi(t_2) =  \frac{\pi\ln(1/(\alpha t_2))}{\ln(1/\eps)} - \frac{\pi\ln(1/t_2)}{\ln(1/\eps)}  = \frac{\pi\ln(1/\alpha)}{\ln(1/\eps)} \leq \frac{2\pi\ln \ln (1/\eps)}{\ln(1/\eps)},
$$
and therefore, we can use the Taylor expansion for cosine and get that 
$$
\alpha(1-\cos(\varphi(t_1)-\varphi(t_2))) \leq \alpha\left(1- \left[1-\frac{\pi^2\ln^2(1/\alpha)}{2\ln^2(1/\eps)}\right]\right) \leq O\left(\frac{\alpha\ln^2(1/\alpha)}{\ln^2(1/\eps)}\right)
$$
which is at most $O(1/\ln^2(1/\eps))$ as $\alpha\ln^2(1/\alpha)$ is at most $e$ for $0\leq \alpha\leq 1$. This completes the proof for this case.
\paragraph{Case 2.} If $t_2\geq 1$ and $\eps \leq |t_1| \leq 1$, then let us again write the term we need to bound
\begin{align*}
\frac{\|g(t_2)-g(t_1)\|^2}{|t_2-t_1|^2} &= \frac{(t_2 - t_1\cos\varphi(t_1))^2 + t_1^2\sin^2\varphi(t_1)}{(t_2-t_1)^2} = \frac{t_1^2 + t_2^2 - 2t_1t_2\cos\varphi(t_1)}{(t_1 - t_2)^2}
\end{align*}
Now if $t_1$ is positive, i.e., $\eps\leq t_1\leq 1$, then clearly, since $\cos \varphi(t_1)\leq 1$, we have that $-2t_1t_2\cos\varphi(t_1)\geq -2t_1t_2$, and therefore the above fraction is at least $1$. Thus, we now consider the case where $-1\leq t_1\leq -\eps$, and need to show that $-2t_1t_2(1-\cos\varphi(t_1))/(t_2-t_1)^2 \leq O(1/\ln^2(1/\eps))$. Again, we let $t_1=-\alpha t_2$ where $0<\alpha \leq 1$, and we get that 
\[
\frac{-2t_1t_2(1-\cos\varphi(t_1))}{(t_2-t_1)^2} = O(\alpha(1-\cos\varphi(t_1)))
\]
Again, if $\alpha\leq 1/\ln^2(1/\eps)$, we have that $\alpha(1-\cos \varphi(t_1)) \leq O(1/\ln^2(1/\eps))$ as $(1-\cos \varphi(t_1))\leq 2$. Otherwise, as $t_2 \geq 1$, we have $|t_1| \geq 1/\ln^2(1/\eps)$, and therefore, $\varphi(t_1) \leq \frac{2\pi\ln\ln(1/\eps)}{\ln(1/\eps)}$. Thus, similar to Case 1, we can write that 
\[
\alpha(1-\cos\varphi(t_1)) \leq \alpha \left(1-\left[1-\frac{\pi^2\ln^2(1/\alpha)}{2\ln^2(1/\eps)}\right]\right) = O(\frac{\alpha\ln^2(1/\alpha)}{\ln^2(1/\eps)}) = O(1/\ln^2(1/\eps))
\]
where the above holds for similar reasons as Case 1.
\paragraph{Case 3.} If $\eps\leq t_2\leq 1$ and $-\eps \leq t_1\leq \eps$, then we have
\begin{align*}
\frac{\|g(t_2)-g(t_1)\|^2}{|t_2-t_1|^2} =\frac{t_1^2+t_2^2+2t_1t_2\cos\varphi(t_2)}{(t_2-t_1)^2} = 1+\frac{2t_1t_2(1+\cos\varphi(t_2))}{(t_2-t_1)^2}
\end{align*}
Now, if $t_1>0$, then the above term is at least $1 \geq 1- O(1/\ln^2(1/\eps))$ and the claim holds. So we assume that $-\eps\leq t_1\leq 0$, and let $t_1 = -\alpha t_2$ where $0\leq \alpha \leq 1$. Our goal is to prove that $-2t_1t_2(1+\cos\varphi(t_2))/(t_2-t_1)^2 \leq O(1/\ln^2(1/\eps))$. Again we can write
\[
\frac{-2t_1t_2(1+\cos\varphi(t_2))}{(t_2-t_1)^2} = O(\alpha(1+\cos\varphi(t_2)))
\]
Now if $\alpha\leq 1/\ln^2(1/\eps)$, we are done as $1+\cos\varphi(t_2)\leq 2$. But then if $\alpha\geq 1/\ln^2(1/\eps)$, we have that $t_2 = -t_1/\alpha \leq \eps/\alpha \leq \eps(\ln^2(1/\eps))$, and therefore, 
\[
\varphi(t_2) \geq \frac{\pi \ln(1/(\eps\ln^2(1/\eps)))}{\ln(1/\eps)} = \pi - \frac{2\pi\ln\ln(1/\eps)}{\ln(1/\eps)}.
\]
Since $\frac{2\pi\ln\ln(1/\eps)}{\ln(1/\eps)}$ is small, we can write the Taylor expansion and get that 
\begin{align*}
\alpha(1+\cos\varphi(t_2)) &= \alpha\left(1+\cos\left(\pi-\frac{2\pi\ln\alpha}{\ln(1/\eps)}\right)\right) \\
&= \alpha\left(1-\cos\left(\frac{2\pi\ln\alpha}{\ln(1/\eps)}\right)\right) \leq \alpha\left(1-\left[1-\left(\frac{2\pi\ln\alpha}{\ln(1/\eps)}\right)^2\right]\right) \\
&\leq O\left(\frac{\alpha\ln^2\alpha}{\ln^2(1/\eps)}\right) = O(\frac{1}{\ln^2(1/\eps)}),
\end{align*}
as desired. This completes the proof of this case.
\end{proof}

Finally, we need to prove the total movement condition as follows\footnote{We remark that a stronger bound for the total movement can be achieved but for our purposes the above bound suffices.}.
\begin{claim}\label{clm:spiral-movement}
For every point $t\in \Re$, one has $\|g(t) - (t, 0)\| \leq O\left(1\right)$.
\end{claim}
\begin{proof}
The claim is clearly true for $|t|\geq 1$. Also for $-\eps\leq t \leq \eps$, the claim holds since those points move by at most $2\eps$. Finally for points that are on the curve, i.e., $\eps \leq |t| \leq 1$, we know that their distances to the origin is preserved. Therefore, by triangle inequality, $\|g(t) - (t, 0)\| \leq 2|t| = O\left(1\right)$.
\end{proof}
This concludes the proof of the lemma.
\end{proof}

%%%%%%%%%%%%%%%%%%%5
\begin{corollary}\label{cor:spirals}
    Let $\eps$ be a sufficiently small constant and let $p,q,p',q'\in \Re$, such that $p' - q' \in (1 \pm \eps) \Delta$, where $\Delta = q - p$ which can be positive or negative.
    Denote
    $\alpha = p - \frac{\Delta}{\eps^{2/3}}$, $\beta = p - \frac{\Delta}{\eps^{1/3}}$,
    $\gamma = q + \frac{\Delta}{\eps^{1/3}}$, $\delta = q + \frac{\Delta}{\eps^{2/3}}$,
    $\alpha' = q' - \frac{\Delta}{\eps^{2/3}}$, $\beta' = q' - \frac{\Delta}{\eps^{1/3}}$,
    $\gamma' = q' + \frac{\Delta}{\eps^{1/3}}$, $\delta' = q' + \frac{\Delta}{\eps^{2/3}}$,
    Then, there exists a map $g \colon \Re \to \Rbb^2$ such that:
    \begin{itemize}
    \item $g(\alpha) = (\alpha', 0)$;
     $g(\beta) = (\gamma', 0)$;
     $g(\gamma) = (\beta', 0)$;
     $g(\delta) = (\delta', 0)$.
    \item $\forall t_1, t_2\in  \Re$, one has $\norm{g(t_1)-g(t_2)} \in \left(1 + O\left(\frac{1}{\log^2(1 / \eps)}\right)\right) \cdot |t_1 - t_2|$.
    \end{itemize}
\end{corollary}
\begin{proof}
Without loss of generality, we will assume that $\Delta$ is positive. Thus we have that $\alpha\leq \beta\leq p\leq q\leq \gamma\leq \delta$, and $\alpha'\leq \beta'\leq q'\leq p'\leq \gamma'\leq \delta'$. The other case is symmetric. Let $m=\frac{p+q}{2}$ which is also equal to $\frac{\alpha+\delta}{2}=\frac{\beta+\gamma}{2}$, and let $m'=\frac{q'+p'}{2}$ which is also equal to $\frac{\alpha'+\delta'}{2}=\frac{\beta'+\gamma'}{2}$.

Let $\eta = \frac{\|p'-q'\|}{\|q-p\|}$ which clearly lies in $\in (1\pm \eps)$.  First we define the map $h\colon \Re\to\Re$  as follows:
$$
h(t)=
\begin{cases}
\alpha'+t-\alpha, &\mbox{if $t\in(-\infty;\beta]$,}\\
\gamma'+t-\gamma, &\mbox{if $t\in[\gamma,\infty)$,}\\
m'+\eta(t-m), &\mbox{otherwise,}\\
\end{cases}
$$
which trivially maps the points from $(-\infty;\beta]$ to $(-\infty;\beta']$, and the points from $[\gamma;\infty)$ to $[\gamma';\infty)$ by translation, and  linearly maps $[\beta,\gamma]$ to $[\beta',\gamma']$ by scaling and translating the points. It is clear that the map is continuous and its distortion is at most $\max\{\eta,1/\eta\}$ which is at most $1+O(\eps)$. 

Now let $g_0\colon \Re\to\Re^2$ be the map of Lemma \ref{lem:basic-map} with $\eps' = \frac{m'-\beta'}{m'-\alpha'} = \frac{\gamma'-m'}{\delta'-m'} =\frac{\Delta(\eta/2+1/\eps^{1/3})}{\Delta(\eta/2+1/\eps^{2/3})}\leq O(\eps^{1/3})$ which has distortion $1+O(1/\log^2(1/\eps')) = 1+O(1/\log^2(1/\eps))$, and define the scale parameter $\lambda = (m'-\alpha') =(\delta'-m')$.
Our final map is just defined as $g(t) = m' + \lambda g_0(\frac{h(t)-m'}{\lambda})$ and it is clear that its distortion $\mathcal{D}_g\leq \mathcal{D}_h\cdot\mathcal{D}_{g_0} \leq (1+\eps)(1+O(1/\log^2(1/\eps))) \leq (1+O(1/\log^2(1/\eps)))$. This proves the second property. 
For the first property we have the following.
\begin{itemize}
\item{$g(\alpha)$:} We have that $h(\alpha) = \alpha'$ and thus $g(\alpha) = m'+(m'-\alpha')g_0(\frac{\alpha'-m'}{m'-\alpha'}) = m'+(m'-\alpha')g_0(-1) = m'-(m'-\alpha') = \alpha'$.

\item{$g(\beta)$:} We have that $h(\beta) = \beta'$ and thus $g(\beta) = m'+(\delta'-m')g_0(\frac{\beta'-m'}{m'-\alpha'}) = m'+(m'-\alpha')g_0(-\eps') = m'+(\delta'-m') \eps'= \gamma'$.

\item{$g(\gamma)$:} We have that $h(\gamma) = \gamma'$ and thus $g(\gamma) = m'+(m'-\alpha')g_0(\frac{\gamma'-m'}{\delta'-m'}) = m'+(m'-\alpha')g_0(\eps') = m'-\eps'(m'-\alpha') = \beta'$.

\item{$g(\delta)$:} We have that $h(\delta) = \delta'$ and thus $g(\delta) = m'+(\delta'-m')g_0(\frac{\delta'-m'}{\delta'-m'}) = m'+(\delta'-m')g_0(1) = m'+(\delta'-m') = \delta'$.
\end{itemize}
\end{proof}

\subsection{Lower bound}\label{sec:line-lower}
In this section, we show that there exist maps with distortion $1 + \varepsilon$ such that every outer extension of it has distortion at least $1 + \Omega\left(1/\log_2 (1/\varepsilon)\right)^2$.

\begin{theorem}[Theorem \ref{thm:lower-bound-on-phi}]
\label{thm:lower-bound-distortion}
There exist $X \subset \Rbb$ and a map $f:X \to \Rbb$ with distortion $1 + O(\varepsilon)$ such that every outer bi-Lipschitz extension
$f': \Rbb \to \Rbb^m$ has distortion at least  $1 + \Omega\left(1/\log^2 (1/\varepsilon))\right)$.
\end{theorem}
\begin{proof}
Consider a map $f$ that maps three points $-\varepsilon$, $0$, and $1$ to points $\varepsilon$, $0$, and $1$, respectively. The map has distortion
$\frac{1+\varepsilon}{1-\varepsilon} = 1 + 2\varepsilon + O(\varepsilon^2)$.
We show that any bi-Lipschitz extension $f':[-\varepsilon,1] \to \Rbb^n$ of $f$ has distortion at least $$1 + \left(\frac{\pi}{2 \log_2 (1/\varepsilon)}\right)^2$$ asymptotically.

Consider  a bi-Lipschitz extension $f':[-\varepsilon,1] \to \Rbb^n$ of $f$. Without loss of generality, we assume that $\varepsilon = 1/2^k$.
Let $x_i = 1/2^i$ for $i\in\{0,1,\dots, k\}$, and $x_i' = f'(x_i)$.
We will need the following claim.
\begin{claim}\label{claim:triangle}
Consider three points $a$, $b$, $c$ on a line such that $b$ lies exactly in the middle between $a$ and $c$; i.e., $ b= (a+c)/2$. Assume that they are mapped to
points $a'$, $b'$, $c'$ in $\Rbb^m$. Let $\alpha$ be the angle between segments $[a', b']$ and $[a', c']$. Then the distortion $D$ of the map is
at least $1/\cos \alpha$ if $\alpha \leq \pi/4$ and $\sqrt{2}$, otherwise.
%\footnote{The former bound is tight; the latter bound can be strengthened.}.
In particular,
$$D \geq \min (1/\cos \alpha, \sqrt{2}).$$
\end{claim}
\begin{proof}
First, assume that $\alpha \leq \pi/4$.
We now show that $\|a'-b'\| \geq \frac{\|a'-c'\|}{2\cos \alpha}$ or $\|b'-c'\| \geq \frac{\|a'-c'\|}{2\cos \alpha}$.
Let $\rho = \|a'-b'\|/ \|a'-c'\|$. If $\rho \geq \frac{1}{2\cos \alpha}$, we are done. Otherwise,
$$\|b'-c'\|^2 = \|a'-b'\|^2 + \|a'-c'\|^2 - 2\cos\alpha \cdot \|a'-b'\| \|a'-c'\| = \|a'-c'\|^2(\rho^2 - 2\cos\alpha \cdot \rho + 1).$$
Now, the polynomial $t^2 - 2\cos\alpha \cdot t + 1$ attains its minimum on $[0,1/(2\cos\alpha)]$ at point $t = 1/(2\cos \alpha)$,
where it equals $1/(2\cos \alpha)^2$
(here we use that $\alpha \leq \pi/4$ and hence $1/(2\cos \alpha) < \cos \alpha$).
Therefore, $\|b'-c'\| \geq \|a'-c'\|/(2\cos \alpha)$, as required.
Note that the distortion is at least
$$\left.\frac{\|a'-b'\|}{\|a'-c'\|} \right/ \frac{\|a-b\|}{\|a-c\|} \qquad \text{and} \qquad
\left.\frac{\|b'-c'\|}{\|a'-c'\|} \right/ \frac{\|b-c\|}{\|a-c\|}.$$
One of these two ratios is at least $1/\cos \alpha$.

Now, assume that $\alpha \in (\pi/4, \pi/2)$.  The distance from $c'$ to the line passing through
$a'$ and $b'$ is $\sin \alpha \|a' -c'\| \geq  \|a' -c'\|/ \sqrt{2}$; in particular, $\|b' - c'\| \geq   \|a' -c'\|/ \sqrt{2}$.
As in the previous case, this implies that the distortion is at least $\sqrt{2}$.
Finally, assume that $\alpha \geq \pi/2$, then the angle at vertex $a'$ in the triangle $a'b'c'$ is obtuse, therefore
$b'c'$ is the longest side of $a'b'c'$. In particular, $\|b' - c'\| \geq   \|a' -c'\|$. We get that the distortion is at least
$2$.
\end{proof}
\begin{figure}
\includegraphics[clip, trim=5cm 9cm 7cm 3.5cm, width=\textwidth]{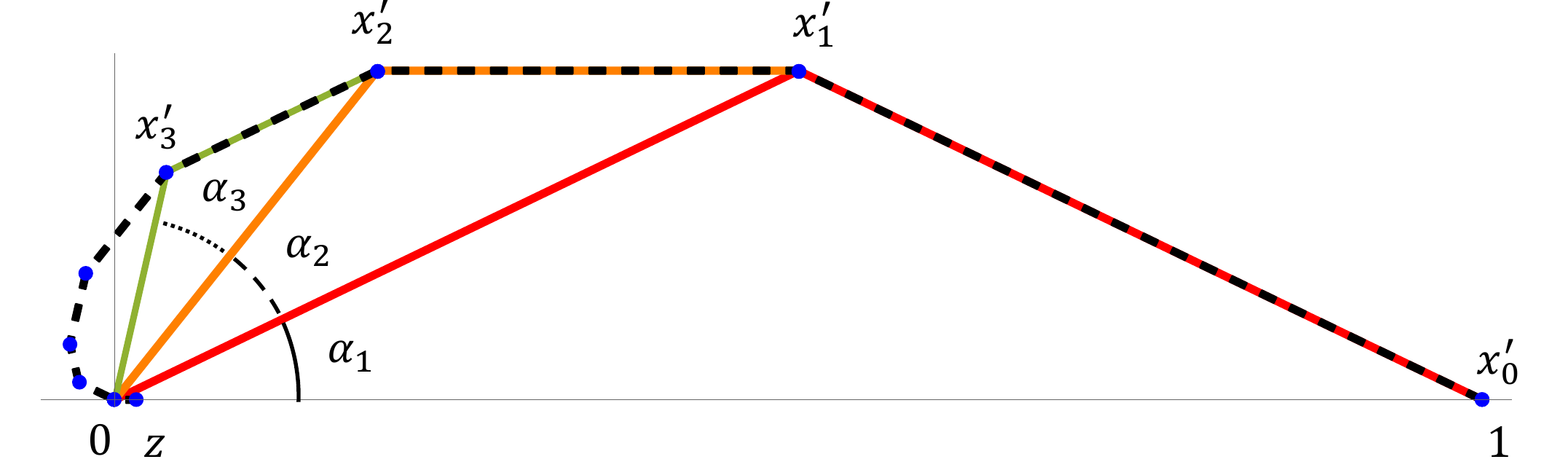}
\caption{Points $x_0',\dots, x_k'$ and angles $\alpha_1, \dots, \alpha_k$.}
\label{fig:points-x-i}
\end{figure}
Now we are ready to prove Theorem~\ref{thm:lower-bound-distortion}.
Let $\alpha_i$ be the angle between segments $[0, x_{i-1}']$ and $[0, x_{i}']$ (see Figure~\ref{fig:points-x-i}). Consider point
$z=(\varepsilon,\bar 0) = f'(-\varepsilon)$. Let $\beta$ be the
largest among the following angles:
\begin{itemize}
\item the angle between $[x_k', z]$ and $[x_k', 0]$,
\item the angle between $[z, x_k']$ and $[z, 0]$.
\end{itemize}
Finally, let $\gamma$ be the angle between
$[0, x_k']$ and $[0,z]$.

First, we apply Claim~\ref{claim:triangle} to points $0$, $x_i$, $x_{i-1}$. We get that
$$D\geq \min\left(\frac{1}{\cos \alpha_i},\sqrt{2}\right).$$
Second, we apply Claim~\ref{claim:triangle} to points $x_k$, $0$, $-\eps$ and to  $-\eps$, $0$, $x_k$
%(recall that $x_k = 1/2^k$ and $\eps = 1/2^k$; $f(x_k) =x_k'$ and $f'(-\eps) = z$; see Figure~\ref{fig:points-xk-z}).
(see Figure~\ref{fig:points-xk-z}).
 We get that
$$D\geq \min\left(\frac{1}{\cos \beta},\sqrt{2}\right).$$
\begin{figure}
\includegraphics[clip, trim=1cm 8.6cm 11cm 6.3cm, width=\textwidth]{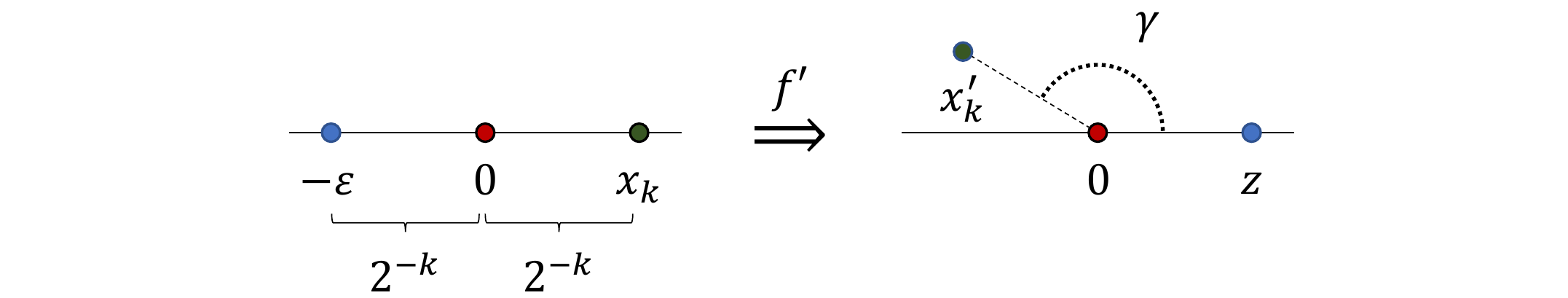}
\caption{Points $-\eps$, $0$, $x_k$ and their images $z = f'(-\eps)$, $0=f'(0)$, $x_k' = f'(x_k)$.}
\label{fig:points-xk-z}
\end{figure}
Now, we write an upper bound for $\gamma$ (which follows from the triangle inequality in spherical geometry)
$$\gamma \leq \sum_{i=1}^k \alpha_i.$$
Consider the triangle with vertices $0$, $z$, $x_k'$. One of the angles of this triangle is $\gamma$ and the largest of the
other two angles is $\beta$. Therefore, $\gamma + 2\beta \geq \pi$ and thus,
$$2\beta + \sum_{i=1}^k \alpha_i \geq \pi.$$
Consequently, either $\beta \geq \pi/(k+2)$ or some $\alpha_i \geq \pi/(k+2)$ (or both). We conclude that the distortion is at least
$$D \geq \min\left(\frac{1}{\cos \frac{\pi}{\log_2 (1/\varepsilon) +2}},\sqrt{2}\right) =  1 + (1 - o(1)) \frac{\pi^2}{2\log_2^2 (1/\varepsilon)}$$
when $\varepsilon \to 0$.
\end{proof}

%............................
\end{document}